\documentclass[11pt]{article}
\usepackage{times}

\usepackage{subcaption}
\usepackage{caption}
\usepackage[paperwidth=199.8mm,paperheight=297mm,centering,hmargin=25mm,vmargin=25mm]{geometry}
\usepackage[dvipsnames]{xcolor}
\usepackage{framed}
\definecolor{shadecolor}{rgb}{0.9,0.9,0.9}
\usepackage{mathtools}
\usepackage{amsmath}
\usepackage[shortlabels]{enumitem}
\usepackage[titletoc,title]{appendix}
\usepackage{graphicx,epic,eepic,epsfig,amsmath,latexsym,amssymb,verbatim,color}
\usepackage{dsfont}

\usepackage{float}
\usepackage{tikz}
\usepackage{hyperref}
\hypersetup{colorlinks=true,citecolor=blue,linkcolor=blue,filecolor=blue,urlcolor=blue,breaklinks=true}

\usepackage{tcolorbox}
\usepackage{relsize}
\usepackage{graphicx}

\usepackage[marginal]{footmisc}
\usepackage{url}

\usepackage{theorem}
\newtheorem{definition}{Definition}
\newtheorem{proposition}[definition]{Proposition}
\newtheorem{lemma}[definition]{Lemma}

\newtheorem{theorem}[definition]{Theorem}
\newtheorem{corollary}[definition]{Corollary}

\usepackage{colortbl}
\usepackage{pifont}
\newcommand{\cmark}{\ding{52}}%

\newcommand{\hcmark}{\cmark\kern-1.1ex\raisebox{.7ex}{\rotatebox[origin=c]{125}{\textbf{--}}}}%
\definecolor{Gray}{gray}{0.92}
\definecolor{Gray2}{gray}{0.75}
\definecolor{maroon}{cmyk}{0,0.87,0.68,0.32}
\usepackage{booktabs} 

\def\squareforqed{\hbox{\rlap{$\sqcap$}$\sqcup$}}
\def\qed{\ifmmode\squareforqed\else{\unskip\nobreak\hfil
\penalty50\hskip1em\null\nobreak\hfil\squareforqed
\parfillskip=0pt\finalhyphendemerits=0\endgraf}\fi}
\def\endenv{\ifmmode\;\else{\unskip\nobreak\hfil
\penalty50\hskip1em\null\nobreak\hfil\;
\parfillskip=0pt\finalhyphendemerits=0\endgraf}\fi}
\newenvironment{proof}{\noindent \textbf{{Proof~} }}{\hfill $\blacksquare$}

\newcounter{remark}
\newenvironment{remark}[1][]{\refstepcounter{remark}\par\medskip\noindent%
\textbf{Remark~\theremark #1} }{\medskip}

\newcounter{example}

\mathchardef\ordinarycolon\mathcode`\:
\mathcode`\:=\string"8000
\def\vcentcolon{\mathrel{\mathop\ordinarycolon}}
\begingroup \catcode`\:=\active
  \lowercase{\endgroup
  \let :\vcentcolon
  }

\newcommand{\nc}{\newcommand}
\nc{\rnc}{\renewcommand}
\nc{\bra}[1]{\langle#1|}
\nc{\ket}[1]{|#1\rangle}
\nc{\<}{\langle}
\rnc{\>}{\rangle}
\nc{\ketbra}[2]{|#1\rangle\!\langle#2|}
\nc{\braket}[2]{\langle#1|#2\rangle}
\nc{\braandket}[3]{\langle #1|#2|#3\rangle}
\nc{\proj}[1]{| #1\rangle\!\langle #1 |}
\nc{\avg}[1]{\langle#1\rangle}
\nc{\Rank}{\operatorname{Rank}}
\nc{\smfrac}[2]{\mbox{$\frac{#1}{#2}$}}
\nc{\tr}{\operatorname{Tr}}
\nc{\ox}{\otimes}

\nc{\cA}{{\cal A}}
\nc{\cB}{{\cal B}}
\nc{\cC}{{\cal C}}
\nc{\cD}{{\cal D}}
\nc{\cE}{{\cal E}}
\nc{\cF}{{\cal F}}
\nc{\cG}{{\cal G}}
\nc{\cH}{{\cal H}}
\nc{\cI}{{\cal I}}
\nc{\cJ}{{\cal J}}
\nc{\cK}{{\cal K}}
\nc{\cL}{{\cal L}}
\nc{\cM}{{\cal M}}
\nc{\cN}{{\cal N}}
\nc{\cO}{{\cal O}}
\nc{\cP}{{\cal P}}
\nc{\cQ}{{\cal Q}}
\nc{\cR}{{\cal R}}
\nc{\cS}{{\cal S}}
\nc{\cT}{{\cal T}}
\nc{\cV}{{\cal V}}
\nc{\cX}{{\cal X}}
\nc{\cY}{{\cal Y}}
\nc{\cZ}{{\cal Z}}
\nc{\cW}{{\cal W}}

\nc{\RR}{{{\mathbb R}}}
\nc{\CC}{{{\mathbb C}}}
\nc{\FF}{{{\mathbb F}}}
\nc{\NN}{{{\mathbb N}}}
\nc{\ZZ}{{{\mathbb Z}}}
\nc{\PP}{{{\mathbb P}}}
\nc{\QQ}{{{\mathbb Q}}}
\nc{\UU}{{{\mathbb U}}}
\nc{\EE}{{{\mathbb E}}}
\nc{\id}{{\operatorname{id}}}

\nc{\1}{{\mathds{1}}}
\nc{\supp}{{\operatorname{supp}}}

\def\ve{\varepsilon}

\newcommand{\eps}{\varepsilon}
\newcommand{\cptp}{\text{\rm CPTP}}

\newcommand{\reg}{{\rm reg}}

\RequirePackage{doi}

\makeatletter
\def\grd@save@target#1{%
  \def\grd@target{#1}}
\def\grd@save@start#1{%
  \def\grd@start{#1}}
\tikzset{
  grid with coordinates/.style={
    to path={%
      \pgfextra{%
        \edef\grd@@target{(\tikztotarget)}%
        \tikz@scan@one@point\grd@save@target\grd@@target\relax
        \edef\grd@@start{(\tikztostart)}%
        \tikz@scan@one@point\grd@save@start\grd@@start\relax
        \draw[minor help lines,magenta] (\tikztostart) grid (\tikztotarget);
        \draw[major help lines] (\tikztostart) grid (\tikztotarget);
        \grd@start
        \pgfmathsetmacro{\grd@xa}{\the\pgf@x/1cm}
        \pgfmathsetmacro{\grd@ya}{\the\pgf@y/1cm}
        \grd@target
        \pgfmathsetmacro{\grd@xb}{\the\pgf@x/1cm}
        \pgfmathsetmacro{\grd@yb}{\the\pgf@y/1cm}
        \pgfmathsetmacro{\grd@xc}{\grd@xa + \pgfkeysvalueof{/tikz/grid with coordinates/major step}}
        \pgfmathsetmacro{\grd@yc}{\grd@ya + \pgfkeysvalueof{/tikz/grid with coordinates/major step}}
        \foreach \x in {\grd@xa,\grd@xc,...,\grd@xb}
        \node[anchor=north] at (\x,\grd@ya) {\pgfmathprintnumber{\x}};
        \foreach \y in {\grd@ya,\grd@yc,...,\grd@yb}
        \node[anchor=east] at (\grd@xa,\y) {\pgfmathprintnumber{\y}};
      }
    }
  },
  minor help lines/.style={
    help lines,
    step=\pgfkeysvalueof{/tikz/grid with coordinates/minor step}
  },
  major help lines/.style={
    help lines,
    line width=\pgfkeysvalueof{/tikz/grid with coordinates/major line width},
    step=\pgfkeysvalueof{/tikz/grid with coordinates/major step}
  },
  grid with coordinates/.cd,
  minor step/.initial=.2,
  major step/.initial=1,
  major line width/.initial=2pt,
}
\makeatother

\newcommand{\Renyi}{R\'{e}nyi }

\def\tD{\widetilde{D}}
\def\td{\widetilde{d}}

\def\tR{\tilde{R}}
\def\ulD{d}

\usepackage{makecell}
\usepackage{diagbox}
\usepackage{multirow}

\def\mE{\mathcal{E}}

\def\mF{\mathcal{F}}
\def\mN{\mathcal{N}}

\def\mS{\mathcal{S}}

\renewcommand{\geq}{\geqslant}
\renewcommand{\leq}{\leqslant}

\def\md{\mathfrak{D}}

\def\ml{\mathfrak{L}}

\newcommand{\bea}{\begin{eqnarray}}
\newcommand{\eea}{\end{eqnarray}}
\newcommand{\be}{\begin{equation}}
\newcommand{\ee}{\end{equation}}
\newcommand{\ba}{\begin{equation}\begin{aligned}}
\newcommand{\ea}{\end{aligned}\end{equation}}

\def\sym{{\rm Sym}}

\def\mi{\mathfrak{I}}
\def\pr{{\rm Pr}}
\def\irr{{\rm Irr}}

\def\be{\begin{equation}}
\def\ee{\end{equation}}

\newcommand{\mR}{\mathcal{R}}
\newcommand{\mM}{\mathcal{M}}

\newcommand{\lr}{\rangle\langle}

\newcommand{\ra}{\rangle}

\newcommand{\mbb}[1]{\mathbb{#1}}

\newcommand{\eqdef}{\coloneqq}

\def\tD{\widetilde{D}}

\def\tR{\tilde{R}}

\def\x{\mathbf{x}}
\def\y{\mathbf{y}}

\def\0{\mathbf{0}}

\def\bd{\bar{d}}

\def\bD{\bar{D}}

\makeatletter
\newcommand*\rel@kern[1]{\kern#1\dimexpr\macc@kerna}
\newcommand*\widebar[1]{%
  \begingroup
  \def\mathaccent##1##2{%
    \rel@kern{0.8}%
    \overline{\rel@kern{-0.8}\macc@nucleus\rel@kern{0.2}}%
    \rel@kern{-0.2}%
  }%
  \macc@depth\@ne
  \let\math@bgroup\@empty \let\math@egroup\macc@set@skewchar
  \mathsurround\z@ \frozen@everymath{\mathgroup\macc@group\relax}%
  \macc@set@skewchar\relax
  \let\mathaccentV\macc@nested@a
  \macc@nested@a\relax111{#1}%
  \endgroup
}
\makeatother

\def\bD{\widebar{D}}
\newcommand{\pro}{\text{\rm PRO}}
\newcommand{\coh}{\text{\rm COH}}
\newcommand{\seq}{\text{\rm SEQ}}

\newcommand{\ST}{\text{\rm E}^{\text{\rm st}}}
\newcommand{\SC}{\text{\rm E}^{\text{\rm sc}}}
\newcommand{\ER}{\text{\rm E}^{\text{\rm er}}}

\newcommand{\boldD}{\boldsymbol{D}}

\DeclareMathOperator*{\Motimes}{\text{\raisebox{0.25ex}{\scalebox{0.6}{$\bigotimes$}}}}

\usepackage{authblk} 

\begin{document}

\title{\Large \textbf{Towards the ultimate limits of quantum channel discrimination\\ and quantum communication}}

\author[1]{Kun Fang~\thanks{kunfang@cuhk.edu.cn}}
\affil[1]{\small School of Data Science, The Chinese University of Hong Kong, Shenzhen, Guangdong, 518172, China}

\author[2]{Gilad Gour~\thanks{giladgour@technion.ac.il}}
\affil[2]{\small Technion - Israel Institute of Technology, Faculty of Mathematics, Haifa 3200003, Israel}   

\author[3]{Xin Wang~\thanks{felixxinwang@hkust-gz.edu.cn}}
\affil[3]{Thrust of Artificial Intelligence, The Hong Kong University of Science and Technology (Guangzhou), Guangdong, 511453, China}

\date{\today}
\maketitle

\vspace{-0.5cm}
\begin{abstract}
Distinguishability is fundamental to information theory and extends naturally to quantum systems. While quantum state discrimination is well understood, quantum channel discrimination remains challenging due to the dynamic nature of channels and the variety of discrimination strategies. This work advances the understanding of quantum channel discrimination and its fundamental limits. We develop new tools for quantum divergences, including sharper bounds on the quantum hypothesis testing relative entropy and additivity results for channel divergences. We establish a quantum Stein's lemma for memoryless channel discrimination, and link the strong converse property to the asymptotic equipartition property and continuity of divergences. Notably, we prove the equivalence of exponentially strong converse properties under coherent and sequential strategies. We further explore the interplay among operational regimes, discrimination strategies, and channel divergences, deriving exponents in various settings and contributing to a unified framework for channel discrimination. Finally, we recast quantum communication tasks as discrimination problems, uncovering deep connections between channel capacities, channel discrimination, and the mathematical structure of channel divergences. These results bridge two core areas of quantum information theory and offer new insights for future exploration.
\end{abstract}

\setcounter{tocdepth}{1}
{
  \hypersetup{linkcolor=black}
  \tableofcontents
}

\newpage

\section{Introduction}

Distinguishability is a central topic in information theory from both theoretical and practical perspectives. A fundamental framework for studying distinguishability is \textit{asymmetric hypothesis testing}. In this setting, a source generates a sample $x$\ from one of two probability
distributions $p\equiv\{p(x)\}_{x\in\mathcal{X}}$ or $q\equiv\{q(x)\}_{x\in
\mathcal{X}}$. The objective of asymmetric hypothesis testing is to minimize the Type-II\ error (decides $p$ when the fact is $q$) while keeping the Type-I\ error (decides $q$ when the fact is $p$) within a certain threshold. The celebrated Chernoff-Stein’s Lemma~\cite{Chernoff1952} states that, for any constant bound on the Type-I error, the optimal Type-II error decays exponentially fast in the number of samples, and the decay rate is exactly the relative entropy (Kullback--Leibler divergence),
\begin{equation}
D(p\Vert q)=\sum_{x\in\mathcal{X}}p(x)\log_{2}[p(x)/q(x)].
\end{equation}
In particular, this lemma also states the ``strong converse property'', a desirable mathematical property in information theory~\cite{Wolfowitz1978} that delineates a sharp boundary for the tradeoff between the Type-I and Type-II errors in the asymptotic regime: any possible scheme with Type-II error decaying to zero with an exponent larger than the relative entropy will result in the Type-I error converging to one in the asymptotic limit.  
Therefore, the Chernoff-Stein’s Lemma provides a rigorous operational interpretation of the relative entropy and establishes a crucial connection between hypothesis testing and information theory~\cite{Blahut1974}.

A natural question is whether the above result generalizes to the quantum case. Substantial efforts have been made to answer this fundamental question in quantum information community~(see, e.g., \cite{hiai1991proper,ogawa2000strong,Hayashi2002,Audenaert2008b,hayashi2007error,brandao2010generalization,cooney2016strong,mosonyi2015quantum,wang2019resource,WW2019,fang2025adversarial,fang2024generalized}). The basic task is quantum state discrimination, in which we are given an independent and identically distributed (i.i.d.) quantum state, which could be either $\rho^{\ox n}$ or $\sigma^{\ox n}$. We set that $\rho^{\ox n}$ is the null hypothesis and $\sigma^{\ox n}$ is the alternative hypothesis. The goal is to perform a binary measurement $\{\Pi_n, I-\Pi_n\}$ on the state to determine which hypothesis is true. The corresponding error probabilities are defined analogously to the classical case, as follows:
\begin{align}
\text{(Type-I)}\quad \alpha_n(\Pi_n):= \tr [(I-\Pi_n) \rho^{\ox n}],\qquad \text{(Type-II)}\quad \beta_n(\Pi_n) := \tr[\Pi_n \sigma^{\ox n}].
\end{align}
The quantum version of the Chernoff-Stein's Lemma (also known as quantum Stein's lemma) states that~\cite{hiai1991proper, ogawa2000strong}
\begin{equation}
\lim_{n \rightarrow \infty}\frac{1}{n} D_H^\varepsilon\left(\rho^{\otimes n} \| \sigma^{\otimes n}\right)=D(\rho \| \sigma), \quad \forall \varepsilon \in(0,1),
\end{equation}
where $D_H^\ve(\rho\|\sigma):= -\log \inf \{\tr [\Pi \sigma]: 0 \leq \Pi \leq I, \tr[(I-\Pi)\rho] \leq \ve\}$ denotes the quantum hypothesis testing relative entropy that characterizes the decay rate of the optimal Type-II error and $D(\rho \| \sigma)= \tr[\rho(\log \rho - \log \sigma)]$ denotes the quantum relative entropy.
This quantum Stein's lemma delivers a rigorous operational interpretation for the quantum relative entropy. Extended research on quantum Stein's lemma are presented in \cite{Hayashi2002,Nagaoka2007,Audenaert2008b,hayashi2007error,brandao2010generalization,Wang2012,fang2025adversarial,fang2024generalized,fang2025efficient}.

Although research in quantum hypothesis testing has
largely focused on quantum states, there are various situations in which quantum channels play the role of the main objects of study. The task of channel discrimination is very similar to that of state discrimination. Here, we are given black-box access to $n$ uses of a channel $\cG$ with the aim to identify it from candidates $\cN$ and $\cM$. 
However, quantum channel discrimination has more diversities in terms of discrimination strategies (e.g., product strategy, coherent strategy, sequential strategy) due to its nature as dynamic resources~\cite{Chiribella2008a,Hayashi2009a,Piani2009b,cooney2016strong,Yuan2017,Berta2018,WW2019,fang2019chain,Pirandola2019}, which leads to several variants of the quantum channel Stein's lemma. 
In particular, for the coherent strategies (also known as parallel strategies in some literatures), the black box can be can be used $n$ times in parallel to any state with a reference system before performing a measurement at the final output to identify the channel. Based on the recently developed resource theory of asymmetric distinguishability for quantum channels, the state-of-the-art result~\cite{WW2019} arrives at
\begin{align}\label{eq:asymp rate of cd}
\lim_{\varepsilon\rightarrow0}\lim_{n\rightarrow\infty}\frac{1}{n}D_{H}^{\varepsilon}(\mathcal{N}^{\otimes n}\Vert\mathcal{M}^{\otimes n}) = D^{\reg}(\cN\|\cM),
\end{align}
with $D_{H}^{\varepsilon}(\mathcal{N}\Vert\mathcal{M})$ denotes the hypothesis testing relative entropy between two quantum channels and
$D^{\reg}(\cN\|\cM)$ denotes the regularized quantum relative entropy.  That is, for Type-I error bounded by $\varepsilon$, the
asymptotic optimal rate of the Type-II error exponent is
given by $D^{\reg}(\cN\|\cM)$ when $\varepsilon$ goes to $0$.

However, the condition of vanishing $\varepsilon$ lefts a notable gap to achieve the quantum channel version of Stein's lemma. 
Unlike state discrimination, the dynamic feature of quantum channels raises challenging difficulties in determining the optimal discrimination scheme as we have to handle the additional optimization of the input states and the non-i.i.d. structure of the testing states. To fill the gap, it requires a deeper understanding and analysis on the error exponent in hypothesis testing of quantum channels. The solution could promptly advance our understanding of quantum channel discrimination, quantum communication~\cite{Wang2017d,Wang2017a,Fang2018,fang2021geometric}, and the related field of quantum metrology~\cite{Pirandola2018a,Braun2018,Degen2017}. Beyond the quantum channel Stein's lemma, various channel divergences have emerged to analyze different regimes of quantum channel discrimination. Establishing a unified framework that encompasses these divergences and discrimination regimes is a desirable step toward a deeper understanding of the manipulation and operational characterization of quantum channels.

In this work, we provide a study towards the utimiate limits of quantum channel discrimination and quantum communication. Our contributions are summerized as follows:
\begin{itemize}
\item In Section~\ref{sec: preliminaries}, we present several technical advancements in quantum divergences for quantum states and channels. Specifically, we provide a quantitative improvement in lower bounding the quantum hypothesis testing relative entropy using the Petz \Renyi divergence, addressing an open question posed by Nuradha and Wilde in~\cite[Remark 4]{nuradha2024fidelity}. Additionally, we demonstrate that the previously explored amortized and regularized channel divergences are generally additive under the tensor product of distinct quantum channels. These technical results are expected to be of independent interest and provide valuable tools for future research.
\item In Section~\ref{sec: Limits of quantum channel divergence}, we investigate the limits of the unstablized quantum channel divergences and prove a quantum channel analog of Stein's lemma without quantum memory assitance. To further strengthen the result, we introduce the (exponetially) strong converse properties for channel discrimination and establish its equivalence to the asymptotic equipartition property (AEP) of various quantum channel divergences as well as the continuity of the quantum channel \Renyi divergence. Leveraging these equivalent characterizations, we demonstrate, rather surprisingly, that the exponentially strong converse properties under coherent and sequential strategies are equivalent. 
\item In Section~\ref{sec: Quantum channel discrimination in different regimes}, we study the interplay between the strategies of channel discrimination (e.g., sequential, coherent, product), the operational regimes (e.g., error exponent, Stein exponent, strong converse exponent), and three variants of channel divergences (e.g., Petz, Umegaki, sandwiched). We find a nice correspondence which shows that the proper divergences to use (Petz, Umegaki, sandwiched) are determined by the operational regime of interest, while the types of channel extension (one-shot, regularized, amortized) are determined by the discrimination strategies. We determine the exponents of quantum channel discrimination in various regimes and contribute towards a complete picture of channel discrimination in a unified framework.

\item In Section~\ref{sec: Quantum communication as quantum channel discrimination}, we present a new perspective by framing the study of quantum communication problems as quantum channel discrimination tasks. This offers deeper insights into the intricate relationships between channel capacities, channel discrimination, and the mathematical properties of quantum channel divergences. Leveraging this connection, we demonstrate that the channel coherent information and quantum channel capacity can be precisely characterized as Stein exponent for discriminating between two quantum channels under product and coherent strategies without quantum memory assistance, respectively. Furthermore, we show that the strong converse property of quantum channel capacity, a long-standing open problem in quantum information theory, can be established if the channels being discriminated exhibit the strong converse property.
\end{itemize} 

Our technical results are primarily presented in terms of the unstabilized channel divergence, a versatile yet less explored notion of channel divergence compared to the more commonly studied divergences in the literature. This concept naturally arises in the context of quantum communication problems, offering a broader framework for analysis. Given the extensive applications of quantum divergences and quantum channel discrimination~\cite{Chitambar2018,Gour2019,WW2019,Fang2020,Regula2021,Wang2017d,berta2018amortized,Wang2017a,fang2019chain,Fang2018,Diaz2018,Saxena2020,WWS19,Faist2018a}, this work contributes to a more comprehensive understanding of the ultimate limits of quantum channel discrimination in various regimes. Moreover, it provides a novel perspective on quantum communication problems by framing them as tasks of quantum channel discrimination, thereby bridging two fundamental areas of quantum information science and paving the way for addressing the remaining challenges in the future studies.

\section{Preliminaries}
\label{sec: preliminaries}

In this section, we introduce the notations to be employed throughout the paper. Subsequently, we investigate the mathematical tool of quantum divergences as applied to both states and channels. Following this, we review the operational task of quantum channel discrimination under different strategies and consider scenarios both with and without quantum memory assistance.

\subsection{Notation}

In this paper, we only consider finite-dimensional Hilbert spaces, which we denote with capital Latin letters such as $C$. The dimension of a Hilbert space $C$ is denoted by $|C|$. The set of linear operators on Hilbert space $C$ is denoted by $\ml(C)$ and the set of density matrices acting on it by $\md(C)$. Density matrices are represented by small Greek letters such as $\rho_C$, where the subscript indicates that $\rho$ acts on $C$. For a state $\rho_{AB}\in\md(AB)$ we will also use the convention that $\rho_A=\tr_B\left[\rho_{AB}\right]$ denotes the marginal on system $A$. The support of an operator $X$ is denoted by $\supp(X)$. The projector onto the subspace spanned by the positive eigenvalues of $X$ is represented as $X_+$. The identity operator is denoted by $I$ and the maximally mixed state is denoted by $\pi$. Quantum channels will be denoted by calligraphic large Latin letters such as $\mN$ and the set of all quantum channels from $A$ to $B$ by CPTP$(A\rightarrow B)$, which stands for completely positive and trace-preserving maps. The identity channel is represented by $\cI$, while the replacer channel is denoted as $\cR^\sigma$, which maps any input state to the fixed state $\sigma$. Throughout the paper we take the logarithm to be base two unless stated otherwise. 

\subsection{Quantum divergences}

A divergence between two quantum states is defined as a real-valued function  $\boldD:\md \times \md \to \RR \cup \{\infty\}$ subject to the data processing inequality $\boldD(\cE(\rho)\|\cE(\sigma)) \leq \boldD(\rho\|\sigma)$ for all quantum states $\rho,\sigma \in \md(A)$ and quantum channel $\cE \in \cptp(A\to B)$. Divergences serve as crucial tools for quantifying the distinguishability of quantum states. In our discussion, we will frequently employ the following quantum divergences, which hold particular relevance.

\begin{definition}[Umegaki relative entropy]
The Umegaki relative entropy (also called quantum relative entropy) between two quantum states $\rho,\sigma \in \md(A)$ is defined by~\cite{umegaki1954conditional}
\begin{align}\label{eq: Umegaki}
  D(\rho\|\sigma):= \tr[\rho(\log \rho - \log \sigma)],
\end{align}
if $\supp(\rho) \subseteq \supp(\sigma)$ and $+\infty$ otherwise.
\end{definition}

\begin{definition}[Petz \Renyi divergence]
The Petz \Renyi divergence of order $\alpha$ between two quantum states $\rho,\sigma \in \md(A)$ is defined by~\cite{petz1986quasi}
\begin{align}\label{eq: Petz}
  \bD_\alpha(\rho\|\sigma)\eqdef \frac{1}{\alpha-1}\log\tr\left[\rho^\alpha\sigma^{1-\alpha}\right],
\end{align}
if $\alpha \in (0,1)$ or $\alpha \in (1,+\infty)$ with $\supp(\rho) \subseteq \supp(\sigma)$, and $+\infty$ otherwise. 
\end{definition}

\begin{definition}[Sandwiched \Renyi divergence]
The sandwiched \Renyi divergence of order $\alpha$ between two quantum states $\rho,\sigma \in \md(A)$ is defined by~\cite{muller2013quantum,wilde2014strong}
\begin{align}\label{eq: sandwiched}
  \tD_{\alpha}(\rho\|\sigma)\eqdef \frac{1}{\alpha-1}\log\tr\left[\sigma^{\frac{1-\alpha}{2\alpha}}\rho\sigma^{\frac{1-\alpha}{2\alpha}}\right]^\alpha,
\end{align}
if $\alpha \in (0,1)$ or $\alpha \in (1,+\infty)$ with $\supp(\rho) \subseteq \supp(\sigma)$, and $+\infty$ otherwise.
\end{definition}

\begin{definition}[Max-relative entropy]
The max-relative entropy between two quantum states $\rho,\sigma \in \md(A)$ is defined by~\cite{datta2009min,renner2005security}
\begin{align}\label{eq: definition of Dmax}
D_{\max}(\rho\|\sigma)\eqdef\log\inf\big\{t \in \RR \;:\; \rho \leq t\sigma \big\}\;,
\end{align}
if $\supp(\rho) \subseteq \supp(\sigma)$ and $+\infty$ otherwise. Let $F(\rho,\sigma) := \|\sqrt{\rho}\sqrt{\sigma}\|_1+\sqrt{(1-\tr\rho)(1-\tr \sigma)}$ be the generalized fidelity and $P(\rho,\sigma):= \sqrt{1-F^2(\rho,\sigma)}$ be the purified distance. Let $\ve \in (0,1)$. Then the smoothed max-relative entropy is defined by
\begin{align}
  D_{\max}^\ve(\rho\|\sigma):= \inf_{\rho':P(\rho',\rho) \leq \ve} D_{\max}(\rho'\|\sigma),
\end{align}
where the infimum is taken over all subnormalized states that are $\ve$-close to the state $\rho$.
\end{definition}

\begin{definition}[Quantum hypothesis testing]
Let $\ve \in [0,1]$. The quantum hypothesis testing relative entropy between two quantum state $\rho,\sigma \in \md(A)$ is defined by
\begin{align}\label{eq: DH}
  D_H^\ve(\rho\|\sigma):= -\log \inf \{\tr [\Pi \sigma]: 0 \leq \Pi \leq I, \tr[\Pi\rho] \geq 1-\ve\}.
\end{align}
\end{definition}

The following result establishes an inequality relating the quantum hypothesis testing relative
entropy and the sandwiched \Renyi divergence~\cite[Lemma 5]{cooney2016strong}.
For any $\alpha \in (1,+\infty)$ and $\varepsilon\in(0,1)$, it holds that
\be\label{eq: DH and sandwiched}
D_H^{\ve}(\rho\|\sigma)\leq \tD_{\alpha}(\rho\|\sigma) + \frac{\alpha}{\alpha-1}\log\frac{1}{1-\varepsilon}.
\ee

The quantum hypothesis testing relative entropy can also be lower bounded by the Petz \Renyi divergence~\cite[Proposition 3]{qi2018applications}. For any $\alpha \in (0,1)$ and $\ve \in (0,1)$, it holds that
\begin{align}
D_H^{\ve}(\rho\|\sigma) \geq \bD_\alpha(\rho\|\sigma) - \frac{\alpha}{1-\alpha} \log \frac{1}{\ve}.
\end{align}

Here we provide a tighter lower bound with a simple proof, addressing the open question posed by Nuradha and Wilde in~\cite[Remark 4]{nuradha2024fidelity}.

\begin{lemma}\label{thm:boundsdmin}
Let $\eps\in(0,1)$ and $\rho,\sigma\in\md(A)$.  
For any $\alpha\in(0,1)$,
\be\label{e2}
D_{H}^\eps(\rho\|\sigma)\geq \bD_\alpha(\rho\|\sigma)+\frac\alpha{1-\alpha}\left(\frac{h(\alpha)}{\alpha}-\log\left(\frac1\eps\right)\right)\;,
\ee
where $h(\alpha)=-\alpha\log\alpha-(1-\alpha)\log(1-\alpha)$ is the binary entropy.
\end{lemma}

\begin{proof}
Recall a variational expression of the hypothesis testing relative entropy~\cite[Eq.(2)]{buscemi2017quantum}
\be\label{2}
2^{-D^{\eps}_{H}\left(\rho\|\sigma\right)}=\max_{t \geq 0}\Big\{t(1-\eps)-\tr(t\rho-\sigma)_+\Big\}\;.
\ee
 To bound the term $\tr(t\rho-\sigma)_+$ we use the quantum weighted geometric-mean inequality; i.e. for any two positive semidefinite matrices $M,N$ and any $\alpha\in[0,1]$
\be\label{7167}
\frac{1}{2}\tr\Big[M+N-\big|M-N\big|\Big]\leq \tr\big[M^{\alpha}N^{1-\alpha}\big]\;.
\ee
Since the term $|M-N|$ can be expressed as $|M-N|=2(M-N)_+-(M-N)$, the above inequality is equivalent to
\be
\tr(M-N)_+\geq \tr[M]-\tr\big[M^{\alpha}N^{1-\alpha}\big].
\ee
Taking $M=t\rho$ and $N=\sigma$ we have
\ba\label{sr1t}
\tr(t\rho-\sigma)_+\geq t-t^{\alpha}\tr\big[\rho^{\alpha}\sigma^{1-\alpha}\big]=t-t^\alpha2^{(\alpha-1)\bD_\alpha(\rho\|\sigma)}\;.
\ea
Substituting this into~\eqref{2} gives
\ba\label{optsti}
2^{-D^{\eps}_{H}\left(\rho\|\sigma\right)}&=\max_{t\geq0}\Big\{t(1-\eps)-\tr(t\rho-\sigma)_+\Big\}\leq \max_{t\geq 0}\Big\{-t\eps+t^\alpha2^{(\alpha-1)\bD_\alpha(\rho\|\sigma)}\Big\}\;.
\ea
It is straightforward to check that for fixed $\alpha,\rho,\sigma,\eps$, the function $t\mapsto -t\eps+t^\alpha2^{(\alpha-1)\bD_\alpha(\rho\|\sigma)}$ obtains its maximal value at 
\be
t=\left(\frac{\alpha}{\eps}\right)^{\frac1{1-\alpha}}2^{-\bD_\alpha(\rho\|\sigma)}\;.
\ee
Substituting this value into the optimization in~\eqref{optsti} gives
\be
2^{-D^{\eps}_{H}\left(\rho\|\sigma\right)}\leq (1-\alpha)\left(\frac{\alpha}\eps\right)^{\frac\alpha{1-\alpha}}2^{-\bD_\alpha(\rho\|\sigma)}\;.
\ee
By taking $-\log$ on both sides we get~\eqref{e2} and conclude the proof.
\end{proof} 
\vspace{0.2cm}

\vspace{0.2cm}

\subsection{Quantum channel divergences}

The divergence between quantum states can be naturally extended to quantum channels. The key idea is to quantify the worst-case divergence among the outputs produced by these channels. Depending on the selection of input states, three distinct variants of quantum channel divergences arise, namely unstabilized, stabilized, and amortized divergences. It is noteworthy that channel divergences have been served as crucial tools in various fundamental areas, including the resource theory of quantum channels~\cite{Chitambar2018,Gour2019,WW2019,Fang2020,Regula2021,theurer2025single}, quantum communication~\cite{Wang2017d,berta2018amortized,Wang2017a,fang2019chain,Fang2018,george2022finite}, quantum coherence~\cite{Diaz2018,Saxena2020}, fault-tolerant quantum computing~\cite{WWS19}, and quantum thermodynamics~\cite{Faist2018a}. We review their definitions here and provide several general properties, which will be used in the later discussions and can be of independent interests for future studies as well.

\subsubsection{Unstabilized quantum channel divergence}

\begin{definition}\label{def: unstabilized channel divergence}
Let $\boldD$ be a quantum state divergence. The unstabilized quantum channel divergence between two quantum channels $\cN,\cM\in \cptp(A\to B)$ is defined by
\begin{align}\label{eq: unstabilized channel divergence}
  \boldsymbol{d}(\cN\|\cM):= \sup_{\rho \in \md(A)} \boldD(\cN_{A\to B}(\rho_A)\|\cM_{A\to B}(\rho_A)),
\end{align}
where the supremum is taken over all density operators $\rho$ on system $A$.
\end{definition}

The term ``unstabilized'' arises from the observation that the divergence value typically varies when appending an identity map, as expressed by the inequality:
\begin{align}
\boldsymbol{d}(\mathcal{N}\|\mathcal{M}) \neq \boldsymbol{d}(\mathcal{N}\otimes \mathcal{I}\|\mathcal{M} \otimes \mathcal{I}).
\end{align}
This distinguishes it from the conventional channel divergence~\cite[Definition II.2]{leditzky2018approaches}.

In the following, we use $d,\td,\bd,d_{\max}^\ve,d_H^\ve$ to represent the unstablized channel divergences induced by $D, \tD, \bD, D_{\max}^\ve, D_H^\ve$, respectively. 

Many properties of state divergence can be extended to channel divergences. For instance, the following continuity property holds true.

\begin{lemma}\label{lem: one shot continuity}
Let $\bar{d}_\alpha, \tilde{d}_\alpha$ and $d$ be the unstablized quantum channel divergences induced by the Petz \Renyi divergence, the sandwiched \Renyi divergence and the Umegaki relative entropy, respectively. Then for any $\mN,\mM\in\cptp(A\to B)$, it holds that
\be
\lim_{\alpha\to 1} \bar{d}_\alpha(\mN\|\mM)=\lim_{\alpha\to 1}\tilde{d}_\alpha(\mN\|\mM)=d(\mN\|\mM)\;.
\ee
\end{lemma}

\begin{proof}
The proof follows similarly as~\cite[Lemma 10]{cooney2016strong}.
\end{proof} \vspace{0.2cm}

A widely-studied unstabilized channel divergence is the min-ouput entropy~\cite{hastings2009superadditivity}
\begin{align}
h(\cN):= \min_{\rho \in \md(A)} H(\cN(\rho)) = \log |B| - \ulD(\cN\|\cR^{\pi}_E),
\end{align}
where $\cN \in \cptp(A\to B)$ and the maximally mixed state $\pi \in \md(B)$.
It is known that this quantity is not additive under tensor product of quantum channels~\cite{hastings2009superadditivity}. 

Given that an unstabilized quantum channel divergence is generally non-additive, it is natural to introduce its regularized counterpart.

\begin{definition} \label{def: regularized unstabilized channel divergence}
Let $\boldD$ be a quantum state divergence. For any $\cN, \cM \in \cptp(A\to B)$, the regularized version of the unstabilized channel divergence is defined by 
 \begin{align}
   \boldsymbol{d}^\reg(\cN\|\cM) := \sup_{n \in \mathbb{N}} \frac1n \boldsymbol{d}(\cN^{\ox n}\|\cM^{\ox n}).
 \end{align}
\end{definition}

If the quantum state divergence $\boldD$ is superadditive under tensor product, i.e., \begin{align}
\boldD(\rho_1\ox \rho_2\|\sigma_1\ox \sigma_2) \geq \boldD(\rho_1\|\sigma_1) + \boldD(\rho_2\|\sigma_2),    
\end{align} 
then it is easy to check that its unstablized channel divergence is also superadditive, i.e., 
\begin{align}
\boldsymbol{d}(\cN_1\ox \cN_2\|\cM_1\ox \cM_2) \geq \boldsymbol{d}(\cN_1\|\cM_1) + \boldsymbol{d}(\cN_2\|\cM_2).   
\end{align}
Using a standard argument, we also have
\begin{align}
    \boldsymbol{d}^\reg(\cN\|\cM) = \lim_{n\to\infty} \frac1n \boldsymbol{d}(\cN^{\ox n}\|\cM^{\ox n}).
\end{align}

Later, as demonstrated in Theorem~\ref{thm: extremely non-additive}, we will see that the unstabilized channel divergence can exhibit an \emph{extremely} non-additive behavior. In other words, an \emph{unbounded} number of channel uses may be necessary to achieve its regularization.

\subsubsection{Stabilized quantum channel divergence}

The unstabilized quantum channel divergence exhibits deviation when an identity map is appended. To mitigate this, we can consider a stabilized version that allows the inclusion of an identity map.

\begin{definition}
Let $\boldD$ be a quantum state divergence. The (stabilized) quantum channel divergence between two quantum channels $\cN,\cM\in \cptp(A\to B)$ is defined by~\cite{leditzky2018approaches}
\begin{align}\label{eq: definition of channel divergence}
  \boldD(\cN\|\cM) := \sup_{|R| \in \mathbb{N}} \boldsymbol{d}(\cI_R \ox \cN\|\cI_R\ox \cM),
\end{align}
where the supremum is taken over Hilbert space $R$ of arbitrary dimension.
\end{definition} 

\begin{remark}
As a consequence of purification, data processing, and the Schmidt decomposition, the supremum can be constrained such that the reference system $R$ is isomorphic to the channel input system $A$~\cite{leditzky2018approaches}. Thus,
$\boldD(\cN\|\cM) = \boldsymbol{d}(\cI_R \ox \cN\|\cI_R\ox \cM)$,
where $R$ is isomorphic to $A$.
\end{remark}

Similar to the unstabilized channel divergence, the stabilized version is non-additive~\cite{fang2019chain} in general. This observation motivates the introduction of their regularization.

\begin{definition}
Let $\boldD$ be a quantum state divergence. For any $\cN, \cM \in \cptp(A\to B)$, the regularized version of the stabilized channel divergence is defined by 
\begin{align}\label{eq: regularized channel divergence}
\boldD^\reg(\mN\|\mM) := \sup_{n \in \mathbb{N}} \frac{1}{n}\, \boldD\left(\mN^{\otimes n}\big\|\mM^{\otimes n}\right).
\end{align} 
\end{definition}

\subsubsection{Amortized quantum channel divergence}

Both the unstabilized and stabilized channel divergences assess the distinguishability of channel outputs using the same input state. Alternatively, a method for inducing channel divergence is amortization, which uses different input states.

\begin{definition}
Let $\boldD$ be a quantum state divergence. The amortized quantum channel divergence between two quantum channels $\cN,\cM\in \cptp(A\to B)$ is defined by~\cite{berta2018amortized}
\begin{align}\label{eq: amortized channel divergence}
\boldD^A(\cN\|\cM):= \sup_{\rho,\sigma \in \md(RA)} \Big[\boldD\left(\cI_R\ox \cN(\rho_{RA})\|\cI_R \ox \cM(\sigma_{RA})\right) - \boldD\left(\rho_{RA}\|\sigma_{RA}\right)\Big],
\end{align}
where the supremum is taken over all quantum states $\rho,\sigma \in \md(RA)$ and $R$ is of arbitrary dimension.
\end{definition}

As previously mentioned, both unstabilized and stabilized channel divergences are generally non-additive. In contrast, the amortized channel divergence can inherit the additivity property from the corresponding state divergence.

\begin{lemma}\label{lem: amortized subadditivity}
Let $\boldD$ be a quantum state divergence. Let $\cN_1,\cM_1 \in \cptp(A_1\to B_1)$ and $\cN_2,\cM_2 \in \cptp(A_2\to B_2)$. If $\boldD$ is additive under tensor product of quantum states, then
  \begin{align}
    \boldD^A(\cN_1\ox \cN_2 \|\cM_1\ox \cM_2) = \boldD^A(\cN_1\|\cM_1) + \boldD^A(\cM_1\|\cM_2).
  \end{align}
\end{lemma}
\begin{proof}
For any quantum state $\rho,\sigma \in \md(RA_1A_2)$, it holds that
\begin{align}
  \boldD(\cN_1\ox \cN_2 (\rho)\|\cM_1\ox \cM_2 (\sigma)) & \leq \boldD^A(\cN_1\|\cM_1) + \boldD(\cN_2(\rho)\|\cM_2(\sigma))\\
  & \leq \boldD^A(\cN_1\|\cM_1) + \boldD^A(\cN_2\|\cM_2) + \boldD(\rho\|\sigma),\label{eq: D amortized subadditive}
\end{align}
where the two inequalities follow by using the definition of the amortized channel divergence twice. Then moving the term $\boldD(\rho\|\sigma)$ to the l.h.s. and taking supremum over all input states $\rho,\sigma$, we have one direction of the stated result. On the other hand, for any input states $\rho_1,\sigma_1 \in \md(R_1A_1)$ and $\rho_2,\sigma_2 \in \md(R_2A_2)$, we have
\begin{align}
& \boldD^A(\cN_1\ox \cN_2 \|\cM_1\ox \cM_2)\\
& \geq \sup_{\rho_1,\rho_2,\sigma_1,\sigma_2} \Big[\boldD(\cN_1\ox \cN_2 (\rho_1\ox \rho_2)\|\cM_1\ox \cM_2 (\sigma_1\ox \sigma_2)) - \boldD(\rho_1\ox \rho_2\|\sigma_1\ox \sigma_2)\Big]\\
& =  \sup_{\rho_1,\rho_2,\sigma_1,\sigma_2} \Big[\boldD(\cN_1(\rho_1)\|\cM_1(\sigma_1) - \boldD(\rho_1\|\sigma_1)\Big] + \Big[\boldD(\cN_2(\rho_2)\|\cM_2(\sigma_2) - \boldD(\rho_2\|\sigma_2)\Big]\\
& =  \sup_{\rho_1,\sigma_1} \Big[\boldD(\cN_1(\rho_1)\|\cM_1(\sigma_1) - \boldD(\rho_1\|\sigma_1)\Big] + \sup_{\rho_2,\sigma_2} \Big[\boldD(\cN_2(\rho_2)\|\cM_2(\sigma_2) - \boldD(\rho_2\|\sigma_2)\Big]\\
& = \boldD^A(\cN_1\|\cM_1) + \boldD^A(\cN_2\|\cM_2),
\end{align}
where the inequality follows as tensor product states are particular choices of input states for $\boldD^A(\cN_1\ox \cN_2 \|\cM_1\ox \cM_2)$, the first equality follows by the additivity assumption of $\boldD$. This concludes the proof.
\end{proof} \vspace{0.2cm}

By the chain rules of Umegaki relative entropy~\cite[Corollary 3]{fang2019chain} and the sandwiched Rényi divergence~\cite[Theorem 5.4]{fawzi2021defining}, it follows that $D^\reg(\cN\|\cM) = D^A(\cN\|\cM)$ and $\tD^\reg_\alpha(\cN\|\cM) = \tD^A_\alpha(\cN\|\cM)$ for any quantum channels $\cN,\cM\in \cptp(A\to B)$ and $\alpha > 1$. Consequently, from Lemma~\ref{lem: amortized subadditivity}, we can infer that $D^\reg$ and $\tD^\reg_\alpha$ are also additive under the tensor product of \emph{distinct} quantum channels. Establishing this directly from their definitions can be challenging.
    
\begin{lemma}
Let $\cN_1,\cM_1 \in \cptp(A_1\to B_1)$ and $\cN_2,\cM_2 \in \cptp(A_2\to B_2)$. For any $\alpha \in (1, +\infty)$, the following additivity properties hold    
\begin{align}
D^\reg(\cN_1\ox \cN_2 \|\cM_1\ox \cM_2) & = D^\reg(\cN_1\|\cM_1) + D^\reg(\cN_2 \|\cM_2),\\
\tD_\alpha^\reg(\cN_1\ox \cN_2 \|\cM_1\ox \cM_2) & = \tD_\alpha^\reg(\cN_1\|\cM_1) + \tD_\alpha^\reg(\cN_2 \|\cM_2).
\end{align}
\end{lemma}

The next result establishes the chain relation among different variants of channel divergences.

\begin{lemma}\label{lem: divergence order relation}
Let $\boldD$ be a quantum state divergence that is superadditive under tensor product of quantum states. Then for any $\mN,\mM\in\cptp(A\to B)$, it holds that
\begin{align}
  \boldsymbol{d}(\cN\|\cM) \leq \boldD(\cN\|\cM) \leq \boldD^\reg(\cN\|\cM) \leq \boldD^A(\cN\|\cM).
\end{align}
\end{lemma}
\begin{proof}
The first two inequalities follow from their definitions. We also have that 
\begin{align}
\frac1n \boldD(\cN^{\ox n}\|\cM^{\ox n}) \leq \frac1n \boldD^A(\cN^{\ox n}\|\cM^{\ox n}) \leq \boldD^A(\cN\|\cM)   
\end{align}
where the first inequality follows by definition and the second inequality follows from~\eqref{eq: D amortized subadditive}. Taking the supremum over all integers $n$, we have $\boldD^\reg(\cN\|\cM) \leq \boldD^A(\cN\|\cM)$.
\end{proof} \vspace{0.2cm}

\subsection{Quantum channel discrimination}

The task of channel discrimination closely parallels that of state discrimination. In the case of an unknown quantum channel $\mathcal{G}$, the goal is to identify it among potential candidates $\mathcal{N}$ or $\mathcal{M}$. A standard approach to discrimination involves hypothesis testing to distinguish between the null hypothesis $\cG = \cN$ and the alternative hypothesis $\cG = \cM$. What distinguishes channel discrimination is the varied selection of discrimination strategies and whether the utilization of quantum memories is permitted. 

Different classes of available strategies are illustrated in Figure~\ref{fig: channel discrimination strategies}.
Each strategy class comprises two components, denoted as $(S_n, \Pi_n)$, where $S_n$ is a method for generating a testing state, and $\Pi_n$ ($0 \leq \Pi_n \leq I$) defines a quantum test, a binary quantum measurement $\{\Pi_n, I - \Pi_n\}$ performed on this state. For a given strategy $(S_n, \Pi_n)$, let $\rho_n(S_n)$ and $\sigma_n(S_n)$ be the testing states generated by $n$ uses of the channel, depending on whether it is $\mathcal{N}$ or $\mathcal{M}$. Then the Type-I and Type-II errors are defined as 
\begin{align}
  \text{(Type-I)} \quad \alpha_n(S_n,\Pi_n)  &:= \tr [(I - \Pi_n) \rho_n(S_n)],\label{eq: definition of two kinds of errors 1} \\ \text{(Type-II)} \quad \beta_n(S_n,\Pi_n)  &:= \tr[\Pi_n \sigma_n(S_n)],\label{eq: definition of two kinds of errors 2}
\end{align}
respectively.
As perfect discrimination (i.e., simultaneous elimination of both errors) is not always possible, the focus shifts to the asymptotic behavior of $\alpha_n$ and $\beta_n$ for sufficiently large $n$, expecting a tradeoff between minimizing $\alpha_n$ and minimizing $\beta_n$.

\begin{figure}[htb]
\centering
\includegraphics[width=\linewidth]{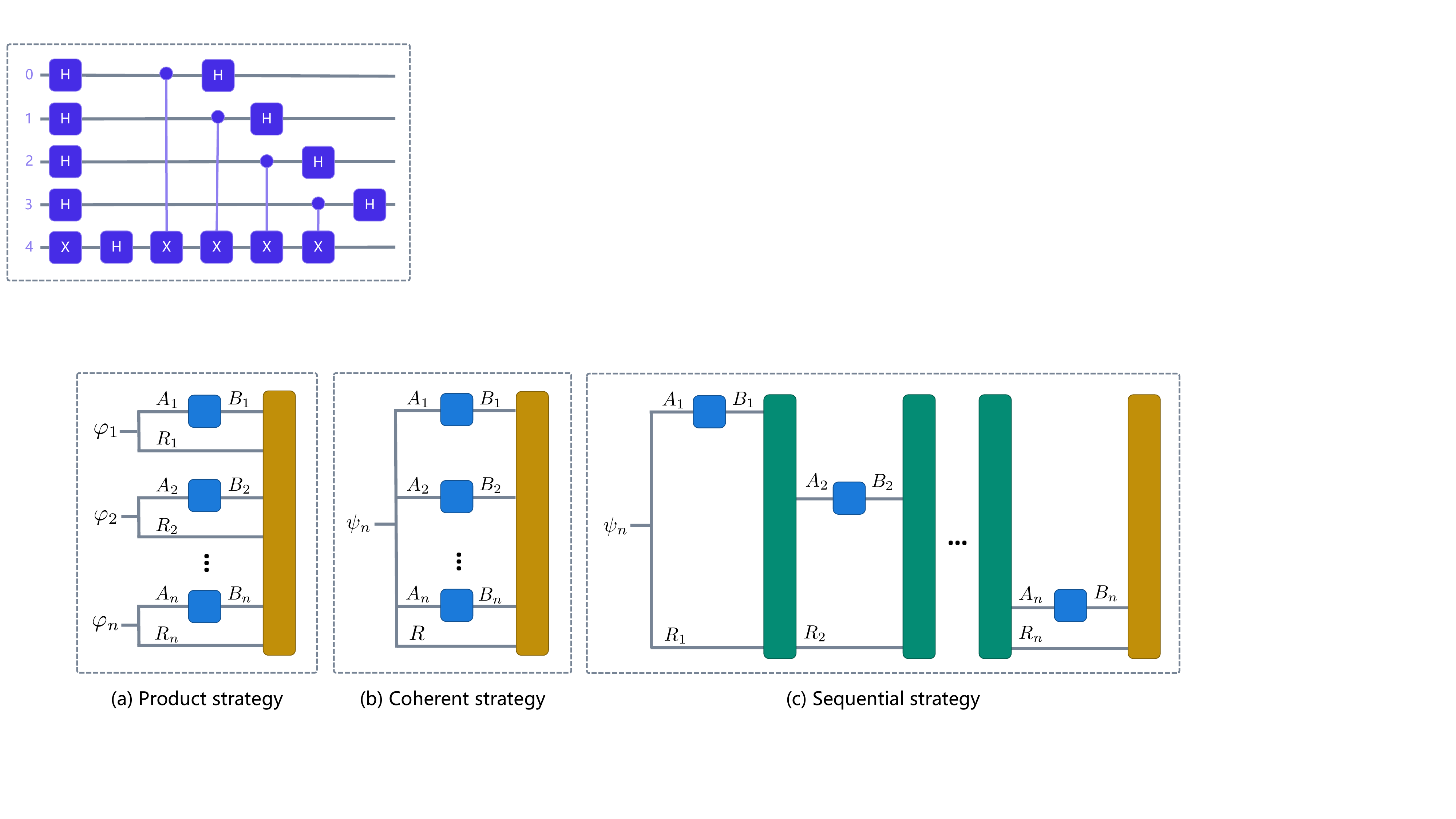}
\caption{Illustration depicting different classes of strategies for quantum channel discrimination. Each blue box represents an unknown quantum channel $\mathcal{G} \in \{\mathcal{N},\mathcal{M}\}$ to discriminate, each yellow box represents a quantum measurement $\{\Pi_n, I - \Pi_n\}$, and each green box represents an update channel $\mathcal{P}_i$.}
\label{fig: channel discrimination strategies}
\end{figure}

\paragraph{Product strategy} 
Let $R_i$ be the ancillary quantum system of a quantum memory for the $i$-th use of the quantum channel. In a product strategy (Figure~\ref{fig: channel discrimination strategies}(a)), the testing state is created by selecting a sequence of input states $\varphi_i \in \md(R_iA_i)$ and sending the $A_i$ system to the unknown channel $\mathcal{G}$ individually. The generated testing state is then given by $\mathcal{G}^{\otimes n}(\bigotimes_{i=1}^n \varphi_i)$. The class of all product strategies is denoted as $\pro$. It is important to note that the input states considered here are not restricted to have an i.i.d. structure (e.g., $\varphi^{\otimes n}$) but rather general tensor product states. In other words, we allow the choice of different input states for different instances of $\mathcal{G}$, distinguishing it from the product strategy discussed in~\cite{cooney2016strong}. If the dimension of the ancillary quantum system reduces to $1$, it corresponds to product strategies without quantum memory assistance.

\paragraph{Coherent strategy} 
Let $R$ be the ancillary quantum system of a quantum memory. In a coherent strategy (Figure~\ref{fig: channel discrimination strategies}(b)), the testing state is created by choosing an input state $\psi_n \in \md(RA^n)$ and sending the corresponding $A_i$ system to each copy of the channel. The generated testing state is then given by $\mathcal{G}^{\otimes n}(\psi_n)$. The class of all coherent strategies is denoted as $\coh$. It is evident that if our choice of $\psi_n$ has a tensor product structure $\bigotimes_{i=1}^n \varphi_i$ with $\varphi_i \in \md(R_iA_i)$, we effectively obtain a product strategy. Thus, we have the set inclusion $\pro \subseteq \coh$. If the dimension of the reference systems reduces to $1$, it corresponds to coherent strategies without quantum memory assistance.

\paragraph{Sequential strategy} 
Let $R_i$ be the ancillary quantum system of a quantum memory for the $i$-th use of the quantum channel. In a sequential strategy (Figure~\ref{fig: channel discrimination strategies}(c)), the testing state is created adaptively. Initially, we choose an initial state $\psi_n \in \md(R_1A_1)$ and send it through one copy of the channel $\mathcal{G}$ followed by the application of an update channel $\mathcal{P}_1$. Subsequently, another copy of the channel $\mathcal{G}$ is applied, followed by an update channel $\mathcal{P}_2$. This process is repeated $n$ times, resulting in the final testing state $\mathcal{G} \circ \mathcal{P}_{n-1} \circ \cdots \circ \mathcal{P}_2 \circ \mathcal{G} \circ \mathcal{P}_1 \circ \mathcal{G} (\psi_n)$, where $\mathcal{P}_i \in \cptp(R_iB_i \to R_{i+1}A_{i+1})$. The class of all sequential strategies is denoted as $\seq$. It is evident that if all update channels $\mathcal{P}_i$ are chosen as identity maps, the sequential strategy reduces to a coherent strategy. Thus, we have $\coh \subseteq \seq$. If the dimension of the ancillary quantum system reduces to $1$, it corresponds to sequential strategies without quantum memory assistance.

\section{Limits of quantum channel divergence}
\label{sec: Limits of quantum channel divergence}

In this section, we investigate the limits of the unstablized quantum channel divergences and prove a quantum channel analog of Stein's lemma without quantum memory assitance. To further strengthen the result, we introduce the (exponetially) strong converse properties for channel discrimination and establish its equivalence to the asymptotic equipartition property (AEP) of various quantum channel divergences as well as the continuity of the quantum channel \Renyi divergence. Leveraging these equivalent characterizations, we demonstrate, rather surprisingly, that the exponentially strong converse properties under coherent and sequential strategies are equivalent. 

Given the widespread applications of quantum Stein's lemma, our channel Stein's lemma is anticipated to have significant implications once its strong converse version is completely solved. Our results contribute to distinct perspectives towards establishing such a result and can serve as building blocks for its applications. This includes facilitating a deeper understanding of the tasks of quantum channel discrimination and quantum communication in subsequent sections.

\subsection{A quantum channel Stein’s lemma without memory assistance}

The following result establishes an analog of the Stein's lemma for quantum channels.

\begin{proposition}\label{prop: channel stein unstablized}
For any two quantum channels $\cN, \cM \in \cptp(A\to B)$, it holds that
\begin{align}
    \lim_{\ve \to 0} \lim_{n\to \infty} \frac{1}{n} \ulD_H^\ve(\cN^{\ox n}\|\cM^{\ox n}) = \ulD^\reg(\cN\|\cM).
\end{align}
\end{proposition}

\begin{proof}
Recall that for any $\rho,\sigma \in \md(A)$ and $\ve \in [0,1)$, it holds that 
\begin{align}
D_H^\ve(\rho\|\sigma) \leq \frac{1}{1-\ve}[D(\rho\|\sigma) + h_2(\ve)]
\end{align} 
where $h_2(\cdot)$ is the binary entropy (see e.g.~\cite{wang2012one}). Applying this to $\cN^{\ox n}(\rho_n)$ and $\cM^{\ox n}(\rho_n)$ and taking supremum over all input states $\rho_n \in \md(A^n)$, we have
\begin{align}
\ulD_H^\ve(\cN^{\ox n}\|\cM^{\ox n}) \leq \frac{1}{1-\ve} \big[{d}(\cN^{\ox n}\|\cM^{\ox n}) + h_2(\ve)\big].
\end{align}
Taking limits on both sides, we have
\begin{align}
 \lim_{\ve \to 0} \lim_{n\to \infty} \frac{1}{n} \ulD_H^\ve(\cN^{\ox n}\|\cM^{\ox n}) \leq \ulD^\reg(\cN\|\cM).  
\end{align}
For the other direction, suppose the optimal solution for $\ulD(\cN\|\cM)$ is taken at $\rho_A$. Then we have
\begin{align}
\lim_{n\to \infty} \frac{1}{n} \ulD_H^\ve(\cN^{\ox n}\|\cM^{\ox n})  \geq   \lim_{n\to \infty} \frac{1}{n} d_H^\ve(\cN(\rho)^{\ox n}\|\cM(\rho)^{\ox n}) = d(\cN(\rho)\|\cM(\rho)) = d(\cN\|\cM),
\end{align}
where the first inequality follows as $\rho^{\ox n}$ is a particular choice for the unstabilized divergence, the first equality follows by the quantum Stein's Lemma, the second equality follows by the optimality assumption of $\rho$. Then for any fixed $m$, by replacing $\cN$ with $\cN^{\ox m}$ and $\cM$ with $\cM^{\ox m}$, we have
\begin{align}
\lim_{n\to \infty} \frac{1}{mn} \ulD_H^\ve(\cN^{\ox mn}\|\cM^{\ox mn})  \geq  \frac{1}{m} \ulD(\cN^{\ox m}\|\cM^{\ox m}).
\end{align}
Finally taking $m \to \infty$ and then $\ve \to 0$, we have the achievable part and conclude the proof.
\end{proof} \vspace{0.2cm}

\subsection{Towards a strong converse version}
\label{sec: Exponentially Strong Converse of Channel Hypothesis Testing}

Similar to the strong converse property of quantum state discimination, an analog property can also be defined for quantum channels.

\begin{definition}[Strong converse property]\label{def: strong converse property}
Let $\cN, \cM \in \cptp(A \to B)$ be two quantum channels. These channels exhibit the strong converse property for coherent channel discrimination strategies without quantum memory assistance if, for any sequence of strategies where the Type-II errors $\beta_n$ satisfy
  \begin{align}\label{eq: Type-II error decays} \liminf_{n \to \infty} -\frac{1}{n} \log \beta_n =: r > d^\reg(\cN \| \cM), \end{align}
  there necessarily exists a subsequence of Type-I errors $\alpha_{n_k}$ that converges to $1$ as $n_k \to \infty$.
\end{definition}

If the strong converse property holds, the Type-I error will typically converge to one exponentially fast. Therefore, we introduce a stronger version by requiring exponential convergence and term this condition as an \emph{exponentially strong converse} property.

\begin{definition}[Exponentially strong converse property]\label{def: exponentially strong converse property}
  Let $\cN, \cM \in \cptp(A \to B)$ be two quantum channels. These channels exhibit the strong converse property for coherent channel discrimination strategies without quantum memory assistance if, for any sequence of strategies where the Type-II errors $\beta_n$ satisfy
  \begin{align}\label{eq: Type-II error expoential decays} \liminf_{n \to \infty} -\frac{1}{n} \log \beta_n =: r > d^\reg(\cN \| \cM), \end{align}
  there necessarily exists a subsequence of Type-I errors $\alpha_{n_k}$ such that  $1-\alpha_{n_k} \leq 2^{-cn_k}$ for a constant $c>0$ and for sufficiently large $n_k$.
\end{definition}

The strong converse properties require the study of all suitable discrimination stategies, which can be hard to validate in general. In the following, we provide several equivalent characterizations related to the limits of unstablized channel divergences.

From the proof of Proposition~\ref{prop: channel stein unstablized}, we actually have a stronger statement that
\begin{align}\label{eq: dH lower bound}
\lim_{n\to\infty}\frac1n d_{H}^\eps\left(\mN^{\otimes n}\big\|\mM^{\otimes n}\right) \geq d^{\reg}(\mN\|\mM), \quad \forall \eps\in (0,1).    
\end{align}

The following result shows that the other direction is equivalent to the strong converse property in Definition~\ref{def: strong converse property}.

\begin{theorem}\label{thm: channel stein lemma}
Let $\mN,\mM\in\cptp(A\to B)$ be two quantum channels. Then these channels exhibit the strong converse property as defined in Definition~\ref{def: strong converse property} if and only if the following holds
\be\label{eq: dH strong AEP}
\limsup_{n\to\infty}\frac1n d_{H}^\eps\left(\mN^{\otimes n}\big\|\mM^{\otimes n}\right) \leq d^{\reg}(\mN\|\mM), \quad \forall \eps\in (0,1).
\ee
\end{theorem}

\begin{proof}
Suppose the strong converse property as defined in Definition~\ref{def: strong converse property} holds and assume that 
\begin{align}\limsup_{n\to\infty}\frac1n d_{H}^\eps\left(\mN^{\otimes n}\big\|\mM^{\otimes n}\right) > d^{\reg}(\mN\|\mM),
\end{align} 
then there exists a subsequence $n_k$ such that $\lim_{n_k\to\infty} \frac{1}{n_k} d_{H}^{\eps}\left(\mN^{\otimes n_k}\big\|\mM^{\otimes n_k}\right) > d^{\reg}(\mN\|\mM)$.
This implies a sequence of strategies such that the Type-I error $\alpha_{n_k} \leq \ve$ and the Type-II error $\lim_{n_k\to \infty} -\frac{1}{n_k} \log \beta_{n_k} > d^\reg(\cN\|\cM)$. 
By Definition~\ref{def: strong converse property}, we know that the second condition implies a subsequence of $\alpha_{n_k}$ converges to $1$, which contradicts to the first condition $\alpha_{n_k} \leq \ve$. So Eq.~\eqref{eq: dH strong AEP} holds. On the other hand, we prove that Eq.~\eqref{eq: dH strong AEP} implies Definition~\ref{def: strong converse property}. For any strategies such that $\liminf_{n\to \infty}-\frac{1}{n}\log \beta_n > d^\reg(\cN\|\cM)$. We now show that there exists a subsequence of $\alpha_n$ converges to $1$. Assume there exists $0 < \ve < 1$ such that $\alpha_n \leq \ve$. By the definition of $d_H^\ve$, we have $-\frac{1}{n}\log \beta_n \leq d_H^\ve(\cN^{\ox n}\|\cM^{\ox n})$. This implies 
\begin{align}
  \liminf_{n\to \infty}-\frac{1}{n}\log \beta_n \leq \limsup_{n\to\infty}\frac1n d_{H}^\eps\left(\mN^{\otimes n}\big\|\mM^{\otimes n}\right) \leq d^{\reg}(\mN\|\mM),
\end{align} 
which forms a contradiction to the assumption that $\liminf_{n\to \infty}-\frac{1}{n}\log \beta_n > d^\reg(\cN\|\cM)$.
\end{proof} \vspace{0.2cm}

\vspace{0.2cm}

The following shows that the AEP of max-relative entropy is also equivalent to Definition~\ref{def: strong converse property}.

\begin{theorem}\label{thm: strong Dmax AEP}
Let $\mN,\mM\in\cptp(A\to B)$ be two quantum channels. Then these channels exhibit the strong converse property as defined in Definition~\ref{def: strong converse property} if and only if the following holds
\be\label{eq: dmax strong AEP}
\limsup_{n\to\infty}\frac1n d_{\max}^\eps\left(\mN^{\otimes n}\big\|\mM^{\otimes n}\right) \leq d^{\reg}(\mN\|\mM), \quad \forall \eps\in (0,1).
\ee
This is also equivalent to the following
\be\label{eq: dmax weak AEP}
\lim_{\ve \to 0}\limsup_{n\to\infty}\frac1n d_{\max}^\eps\left(\mN^{\otimes n}\big\|\mM^{\otimes n}\right) \leq d^{\reg}(\mN\|\mM).
\ee
\end{theorem}

\begin{proof}
By Theorem~\ref{thm: channel stein lemma}, we only need to prove that Eqs.~\eqref{eq: dH strong AEP},~\eqref{eq: dmax strong AEP} and~\eqref{eq: dmax weak AEP} are equivalent. 

(i) Eq.~\eqref{eq: dH strong AEP} $\implies$ Eq.~\eqref{eq: dmax strong AEP}: For any two quantum states $\rho,\sigma\in\md(A)$, any $\eps\in(0,1)$,  
it is known that~\cite[Proposition 4.1]{dupuis2014generalized},
\be
D_{\max}^{\eps}\left(\rho\|\sigma\right)\leq D_{H}^{1-\frac{1}{2}\eps^2}(\rho\|\sigma)+\log\left(\frac{2}{\eps^2}\right)\;.
\ee
Applying this to channel divergence gives
\be
d_{\max}^{\eps}\left(\mN\|\mM\right)\leq d_{H}^{1-\frac{1}{2}\eps^2}(\mN\|\mM)+\log\left(\frac{2}{\eps^2}\right).
\ee
Taking $n$ copies of $\mN$ and $\mM$ we get
\begin{align}
\frac1n d_{\max}^{\eps}\left(\mN^{\otimes n}\big\|\mM^{\otimes n}\right)\leq \frac1n d_{H}^{1-\frac{1}{2}\eps^2}\left(\mN^{\otimes n}\big\|\mM^{\otimes n}\right)+\frac1n\log\left(\frac{2}{\eps^2}\right)\;.\label{sandw}
\end{align}
Taking $\limsup_{n\to \infty}$ on both sides, we can see that Eq.~\eqref{eq: dH strong AEP} implies Eq.~\eqref{eq: dmax strong AEP}.

(ii) Eq.~\eqref{eq: dmax strong AEP} $\implies$ Eq.~\eqref{eq: dmax weak AEP}: trivial.

(iii) Eq.~\eqref{eq: dmax weak AEP} $\implies$ Eq.~\eqref{eq: dH strong AEP}: For any two quantum states $\rho,\sigma\in\md(A)$, any $\eps\in(0,1)$, and any $\eps'\in(0,1-\eps)$, it is known that~\cite[Theorem 11]{datta2013smooth},
\begin{align}\label{eq: Dh and Dmax 1}
D_{H}^{\eps'}(\rho\|\sigma)+\log\left(1-\eps-\eps'\right)\leq D_{\max}^{\eps}\left(\rho\|\sigma\right).
\end{align}
Applying this to channel divergence gives
\begin{align}\label{eq: Dh and Dmax}
d_{H}^{\eps'}(\mN\|\mM)+\log\left(1-\eps-\eps'\right)\leq d_{\max}^{\eps}\left(\mN\|\mM\right).
\end{align}
Taking $n$ copies of $\mN$ and $\mM$ we get
\begin{align}
\frac1n d_{H}^{\eps'}\left(\mN^{\otimes n}\big\|\mM^{\otimes n}\right)+\frac1n \log\left(1-\eps-\eps'\right) \leq \frac1n d_{\max}^{\eps}\left(\mN^{\otimes n}\big\|\mM^{\otimes n}\right).    
\end{align}
Taking $\limsup_{n\to \infty}$ and $\lim_{\ve \to 0}$, we get 
\begin{align}
\frac1n d_{H}^{\eps'}\left(\mN^{\otimes n}\big\|\mM^{\otimes n}\right) \leq d^{\reg}(\mN\|\mM),    
\end{align}
which implies Eq.~\eqref{eq: dH strong AEP}.
\end{proof} \vspace{0.2cm}

It is interesting to see that Eq.~\eqref{eq: dmax strong AEP} and Eq.~\eqref{eq: dmax weak AEP} are actually equivalent, despite the latter appearing much weaker than the former. As $D_H$ and $D_{\max}$ are the two extreme cases of one-shot quantum divergences, the above result would also apply to other intermediate divergences such as the information spectrum relative entropies~\cite{tomamichel2013hierarchy,datta2014second} and the recently introduced smoothed sandwiched \Renyi divergence~\cite{rubboli2022fundamental}.

\vspace{0.2cm}

Besides the above AEPs, the strong converse properties also relate to the continuity of the regularized (amortized) sandwiched \Renyi channel divergence at $\alpha = 1$.

\begin{theorem}\label{thm: sandwiched channel continuity regularized}
Let $\mN,\mM\in\cptp(A\to B)$ be two quantum channels. Then the following continuitity
\begin{align}\label{eq: sandwiched continuitity}
  \lim_{\alpha \to 1^+} {\td}_{\alpha}^{\reg}(\cN\|\cM) = d^{\reg}(\cN\|\cM).
\end{align}
implies that these channels exhibit the exponentially strong converse property as defined in Definition~\ref{def: exponentially strong converse property}. Conversely, if the exponentially strong converse property as defined in Definition~\ref{def: exponentially strong converse property} holds true for channels $\cI \ox \cN$ and $\cI \ox \cM$ with the identity channel $\cI \in \cptp(A\to A)$, then 
\begin{align}\label{eq: sandwiched continuitity stablized}
  \lim_{\alpha \to 1^+} {\td}_{\alpha}^{\reg}(\cI \ox \cN\|\cI \ox \cM) = d^{\reg}(\cI \ox \cN\|\cI \ox \cM).
\end{align}
\end{theorem}

\begin{proof}
Note that by the monotonicity of sandwiched \Renyi divergence with respect to $\alpha$, the limits in the above statement can be replaced with $\inf_{\alpha > 1}$. Suppose the continuity in Eq.~\eqref{eq: sandwiched continuitity} holds. Recall that~\cite[Lemma 5]{cooney2016strong} for any $\alpha \in (1,+\infty)$ and $\varepsilon\in(0,1)$, it holds that
\be
D_H^{\ve}(\rho\|\sigma)\leq \tD_{\alpha}(\rho\|\sigma) + \frac{\alpha}{\alpha-1}\log\frac{1}{1-\varepsilon}.
\ee
Applying this to the discimination of $n$ copies of the channels, it implies 
\begin{align}
  -\frac{1}{n} \log \beta_n \leq \frac{1}{n} \td_\alpha(\cN^{\ox n}\|\cM^{\ox n}) + \frac{1}{n} \frac{\alpha}{\alpha-1}\log\frac{1}{1-\alpha_n}.
\end{align}
If $\liminf_{n\to \infty}-\frac{1}{n} \log \beta_n := r > d^\reg(\cN\|\cM)$, then there exists a subsequence $n_k$ and $\delta > 0$ such that $-\frac{1}{n_k} \log \beta_{n_k} > r - \delta > d^\reg(\cN\|\cM)$. Let $r':= r -\delta$. We have 
\begin{align}
  r' < \frac{1}{n_k} \td_\alpha(\cN^{\ox n_k}\|\cM^{\ox n_k}) + \frac{1}{n_k} \frac{\alpha}{\alpha-1}\log\frac{1}{1-\alpha_{n_k}}.
\end{align}
Since $\frac{1}{n_k} \td_\alpha(\cN^{\ox n_k}\|\cM^{\ox n_k}) \leq \td_\alpha^\reg(\cN\|\cM)$, we have 
\begin{align}
r' < \td_\alpha^\reg(\cN\|\cM) + \frac{1}{n_k} \frac{\alpha}{\alpha-1}\log\frac{1}{1-\alpha_{n_k}},
\end{align}
which is equivalent to 
\begin{align}
  1-\alpha_{n_k} < 2^{-\frac{\alpha-1}{\alpha}n_k\left(r'-\td_\alpha^\reg(\cN\|\cM)\right)}.
\end{align}
Since $r'>d^\reg(\cN\|\cM)=\inf_{\alpha>1} {\td}_{\alpha}^{\reg}(\cN\|\cM)$ by assumption, there exists $\alpha > 1$ such that $r' > {\td}_{\alpha}^{\reg}(\cN\|\cM)$. We can choose $c:= (\alpha-1)/\alpha (r' - {\td}_{\alpha}^{\reg}(\cN\|\cM))$. This implies the expoentially strong converse property in Definition~\ref{def: exponentially strong converse property}.

We now prove the second statement. Suppose the expoentially strong converse property in Definition~\ref{def: exponentially strong converse property} holds true for $\cI \ox \cN$ and $\cI \ox \cM$.
For any $\alpha > 1$, we have $\tD_{\alpha}^{\reg}(\cN\|\cM) \geq D^{\reg}(\cN\|\cM)$. Thus it is clear that $\inf_{\alpha > 1} \tD_{\alpha}^{\reg}(\cN\|\cM) \geq D^{\reg}(\cN\|\cM)$. We now prove the other direction. If $\inf_{\alpha > 1} \tD_{\alpha}^{\reg}(\cN\|\cM) > D^{\reg}(\cN\|\cM)$, we can find  $r \in \RR $ such that $\inf_{\alpha > 1} \tD_{\alpha}^{\reg}(\cN\|\cM) > r > D^{\reg}(\cN\|\cM)$. Consider a sequence of coherent channel discrimination strategies such that the Type-II error converges at an exponential rate $r$. By the result~\cite[Theorem 5.5 and Remark 5.6]{fawzi2021defining}, we know that the strong converse exponent is zero since $r < \inf_{\alpha > 1} \tD_{\alpha}^{\reg}(\cN\|\cM)$, which means the Type-I error does not exponentially converge to one. However, by Definition~\ref{def: exponentially strong converse property}, the condition $r > D^{\reg}(\cN\|\cM)$ implies that the Type-I error has to converge exponentially to one, which forms a contradiction and concludes that $\inf_{\alpha > 1} \tD_{\alpha}^{\reg}(\cN\|\cM) \leq D^{\reg}(\cN\|\cM)$. This proves Eq.~\eqref{eq: sandwiched continuitity stablized}. 
\end{proof} \vspace{0.2cm}

\vspace{0.2cm}
Note that the second statement above holds for the stabilized channel divergence, as its proof relies on the results in~\cite[Theorem 5.5 and Remark 5.6]{fawzi2021defining}. It would be interesting to determine whether this result also holds for the unstabilized channel divergence in general.

\begin{corollary}\label{thm: sandwiched channel continuity amortized}
Let $\cN,\cM \in \cptp(A\to B)$ be two quantum channels. The exponentially strong converse property as defined in Definition~\ref{def: exponentially strong converse property} holds true for channels $\cI \ox \cN$ and $\cI \ox \cM$ with the identity channel $\cI \in \cptp(A\to A)$ if and only if one of the following continuities hold  
\begin{align}
  \lim_{\alpha \to 1^+} \tD_{\alpha}^{\reg}(\cN\|\cM) & = D^{\reg}(\cN\|\cM),\\
    \lim_{\alpha \to 1^+} \tD_{\alpha}^{A}(\cN\|\cM) & = D^{A}(\cN\|\cM).
\end{align}
\end{corollary}

\begin{proof}
The first equation follows from Theorem~\ref{thm: sandwiched channel continuity regularized}. The second equation follows from the existing results $\tD_{\alpha}^{\reg}(\cN\|\cM) = \tD_{\alpha}^{A}(\cN\|\cM)$~\cite[Theorem 5.4]{fawzi2021defining} and $D^{\reg}(\cN\|\cM) = D^{A}(\cN\|\cM)$ \cite[Corollary 3]{fang2019chain}.
\end{proof} \vspace{0.2cm}

Note that the exponentially strong converse property in Definition~\ref{def: exponentially strong converse property} is defined for coherent strategies. Here, we demonstrate that it is equivalent to the exponentially strong converse property under sequential strategies. This is quite remarkable, as sequential strategies can be significantly more general than coherent strategies.

\begin{theorem}\label{coro: strong converse of channel hypothesis testing sequential strategies}
Let $\cN,\cM \in \cptp(A\to B)$ be two quantum channels and $\cI \in \cptp(A \to A)$ be the identity channel. The expoentially strong converse property in Definition~\ref{def: exponentially strong converse property} holds true under coherent strategies for channels $\cI \ox \cN$ and $\cI \ox \cM$ if and only if it holds true under squential strategies. 
\end{theorem}
\begin{proof}
By the result in~\cite[Proposition 20]{Berta2018}, for any sequential strategies and $\alpha > 1$ it holds that
\begin{align}\label{eq: strong converse of channel hypothesis testing sequential strategies tmp1}
   -\frac{1}{n} \log(1-\alpha_n) \geq  \frac{\alpha-1}{\alpha} \left(-\frac{1}{n} \log \beta_n - \tD_\alpha^A(\cN\|\cM)\right).
\end{align}
By the expoentially strong converse property under sequential strategies, we assume the relation that $\liminf_{n\to \infty} -\frac{1}{n} \beta_n :=r > D^\reg(\cN\|\cM) = D^A(\cN\|\cM)$ where the second equality follows by~\cite[Corollary 3]{fang2019chain}. This implies that there exists $\delta > 0$ and a subsequence $\beta_{n_k} $ such that
$-\frac{1}{n_k} \log \beta_{n_k} \geq r - \delta > D^A(\cN\|\cM)$ 
for sufficiently large $n_k$. By Corollary~\ref{thm: sandwiched channel continuity amortized}, the expoentially strong converse property in Definition~\ref{def: exponentially strong converse property} is equivalent to the continuity of the amortized channel divergence $\lim_{\alpha \to 1^+} \tD_{\alpha}^A(\cN\|\cM)= D^{A}(\cN\|\cM)$. 
As $r > D^A(\cN\|\cM)$, there exists $\alpha_0 > 1$ such that $r - \delta > \tD^A_{\alpha_0}(\cN\|\cM)$. Then we have
\begin{align}
   -\frac{1}{n_k} \log \beta_{n_k} - \tD^A_{\alpha_0}(\cN\|\cM) \geq r - \delta - \tD^A_{\alpha_0}(\cN\|\cM) =: b > 0.
\end{align}
Taking this into Eq.~\eqref{eq: strong converse of channel hypothesis testing sequential strategies tmp1}, we get
\begin{align}
   -\frac{1}{n_k} \log(1-\alpha_{n_k}) \geq \frac{\alpha_0 - 1}{\alpha_0} b =: c > 0,
\end{align}
which is equivalent to $1-\alpha_{n_k} \leq 2^{-cn_k}$. This establishes the exponentially strong converse property under sequential strategies. Conversely, since any coherent strategy is a specific case of a sequential strategy, the strong converse property under sequential strategies also implies the property under coherent strategies.
\end{proof} \vspace{0.2cm}

\section{Quantum channel discrimination in different regimes}
\label{sec: Quantum channel discrimination in different regimes}

The task of channel discrimination aims to distinguish a quantum channel from the other under a given type of strategy. A standard approach for discrimination is to perform hypothesis testing and make a decision based on the testing result. However, two types of error (Type-I error and Type-II error) arise. In the same spirit of state discrimination, one can study the asymptotic behavior of these errors in different operational regimes (see Figure~\ref{fig:discrimination regimes}), particularly, (I) error exponent regime that studies the exponent of the exponential convergence of the Type-I error given that the Type-II error exponentially decays ; (II) Stein exponent regime that studies the exponent of the exponential decay of the Type-II error given that the Type-I error is within a constant threshold; (III) strong converse exponent regime that studies the exponent of the exponential convergence of the Type-I error given that the Type-II error exponentially decays.

\begin{figure}[htb]
    \centering
    \includegraphics[width=0.8\textwidth]{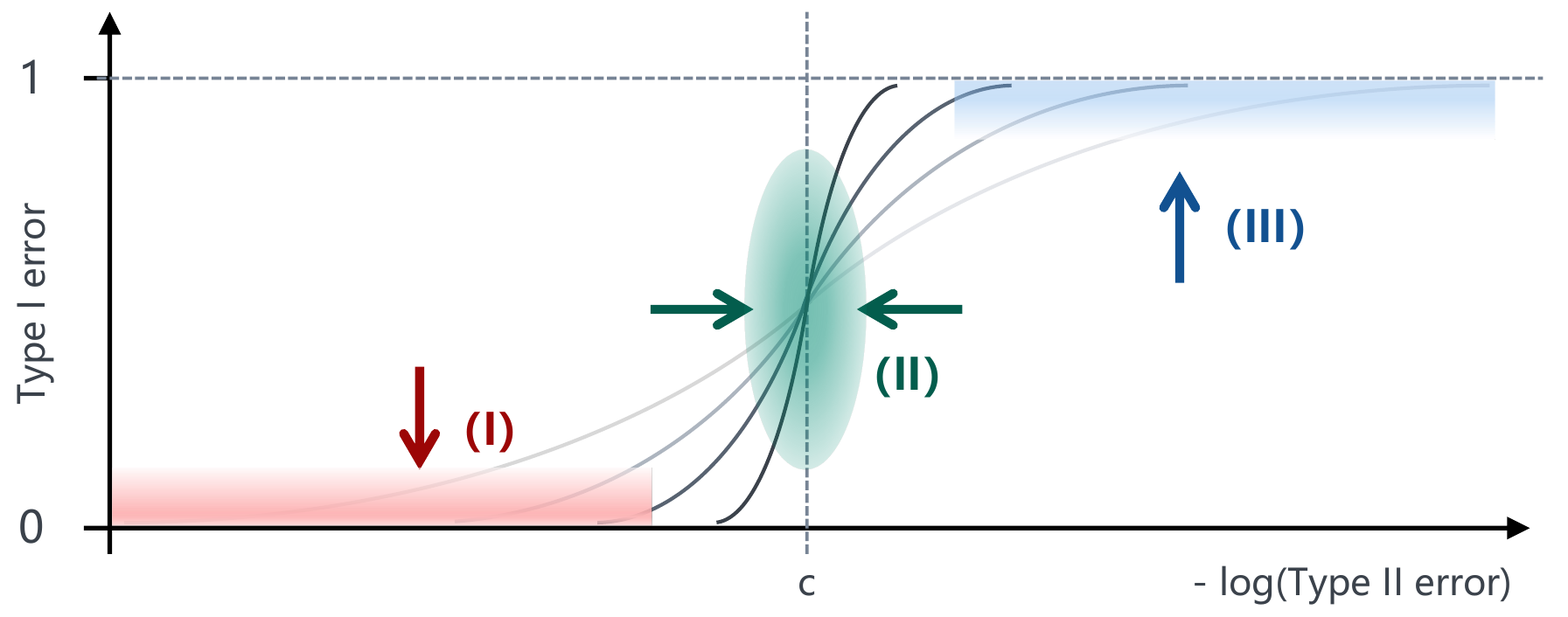}
    \caption{Illustration depicting different regimes of quantum channel discrimination. Each curve represents the tradeoff between the Type-I and Type-II errors for varying block lengths, with darker lines corresponding to longer block lengths. (I) represents the error exponent regime, (II) represents the Stein exponent regime, and (III) represents the strong converse exponent regime.}
    \label{fig:discrimination regimes}
\end{figure}

Quantum state discrimination in different operational regimes has been well-studied. In particular, there is a nice correspondence between the regime studied and the quantum divergence to use. More precisely, the Stein exponent is given by the Umegaki relative entropy~\cite{hiai1991proper,ogawa2000strong}, the strong converse exponent is determined by the sandwiched \Renyi divergence~\cite{mosonyi2015quantum}, and the error exponent is determined by the Petz \Renyi divergence~\cite{hayashi2007error,nagaoka2006converse,audenaert2008asymptotic}. However, when it comes to channel discrimination, the situation becomes much involved due to the diverse range of discrimination strategies and different extensions of channel divergence.  
 
In this section we study the interplay between the strategies of channel discrimination (e.g., sequential, coherent, product), the operational regimes (e.g., error exponent, Stein exponent, strong converse exponent), and three variants of channel divergences (e.g., Petz, Umegaki, sandwiched). We find a nice correspondence which shows that the proper divergences to use (Petz, Umegaki, sandwiched) are determined by the operational regime of interest, while the types of channel extension (one-shot, regularized, amortized) are determined by the discrimination strategies. Our results contribute towards a complete picture of channel discrimination in a unified framework.

\subsection{Stein exponent}

In this subsection we consider minimizing the Type-II error probability, under the constraint that the Type-I error probability does not exceed a constant threshold $\ve \in (0,1)$. We characterize the exact exponent, named \emph{Stein exponent}, with which the Type-II error exponentially decays.

\begin{definition}[Stein exponent]\label{def: stein exponent}
Let $\cN,\cM \in \cptp(A\to B)$ be two quantum channels and $\ve \in (0,1)$ be a fixed error. The Stein exponents of quantum channel discrimination by the strategy class $\Omega \in \{\pro,\coh,\seq\}$ without quantum memory assistance are defined by
\begin{align}\label{eq: stein exponent definition}
\ST_{\Omega,\sup}(\ve|\cN\|\cM)&:= \limsup_{n\to \infty} \frac1n d_{H,\Omega}^\ve(\cN^{\ox n}\|\cM^{\ox n}),\\
\ST_{\Omega,\inf}(\ve|\cN\|\cM)&:= \liminf_{n\to \infty} \frac1n d_{H,\Omega}^\ve(\cN^{\ox n}\|\cM^{\ox n}) ,
\end{align}
where 
\begin{align}\label{eq: DH omega}
d_{H,\Omega}^\ve(\cN^{\ox n}\|\cM^{\ox n}) := \sup_{(S_n,\Pi_n)\in \Omega} \left\{ - \frac1n \log \beta_n(S_n,\Pi_n): \alpha_n(S_n,\Pi_n) \leq \ve \right\},
\end{align}
the supremum is taken over all possible strategies $(S_n,\Pi_n) \in \Omega$ satisfying the condition and the type-I and type-II errors are defined in~Eqs.~\eqref{eq: definition of two kinds of errors 1} and~\eqref{eq: definition of two kinds of errors 2}, respectively.
\end{definition}

The non-asymptotic quantity in~\eqref{eq: DH omega} can also be written as a notion of hypothesis testing relative entropy between the testing states, 
\begin{align}
d_{H,\Omega}^\ve(\cN^{\ox n}\|\cM^{\ox n}) = \sup_{S_n \in \Omega} D_H^\ve(\rho_n(S_n)\|\sigma_n(S_n)),
\end{align}
where the hypothesis testing relative entropy on the r.h.s. is between two quantum states which is defined in~\eqref{eq: DH} and the supremum is taken over all strategies $S_n \in \Omega$ that generate the testing states $\rho_n(S_n)$ and $\sigma_n(S_n)$. More explicitly, when $\Omega = \pro$, we have $\rho_n(S_n) = \cN^{\ox n}(\Motimes_{i=1}^n \varphi_i)$, $\sigma_n(S_n) = \cM^{\ox n}(\Motimes_{i=1}^n \varphi_i)$ and the supremum is taken over all $\varphi_i \in \md(R_iA_i)$. When $\Omega = \coh$, we have $\rho_n(S_n) = \cN^{\ox n}(\psi_n)$, $\sigma_n(S_n) = \cM^{\ox n}(\psi_n)$ and the supremum is taken over all $\psi_n \in \md(RA^n)$. When $\Omega = \seq$, we have $\rho_n(S_n) = \cN\circ \cP_{n-1} \circ \cdots \circ \cP_2 \circ \cN \circ \cP_1 \circ \cN (\psi_n)$, $\sigma_n(S_n) = \cM\circ \cP_{n-1} \circ \cdots \circ \cP_2 \circ \cM \circ \cP_1 \circ \cM (\psi_n)$ and the supremum is taken over all $\psi_n \in \md(R_1A_1)$ and $\cP_i \in \cptp(R_iB_i \to R_{i+1}A_{i+1})$.

\begin{theorem}[Product strategy]\label{thm: stein exponent prodcut}
  Let $\mN,\mM\in\cptp(A\to B)$ be two quantum channels and $\varepsilon\in(0,1)$ be a fixed error. Then it holds that
  \begin{align}
    \ST_{\pro, \sup}(\ve|\cN\|\cM) = \ST_{\pro, \inf}(\ve|\cN\|\cM) =  d(\cN\|\cM).
  \end{align}
\end{theorem}

\begin{proof}
It suffices to show that 
\begin{align}
\lim_{n\to \infty} \frac1n d_{H,\pro}^\ve(\cN^{\ox n}\|\cM^{\ox n}) = d(\cN\|\cM).
\end{align} 
For the achievable part, let $\varphi \in \md(A)$ be an optimal input state for $d(\cN\|\cM)$, i.e., $D(\cN(\varphi)\|\cM(\varphi)) = d(\cN\|\cM)$.  Using $\varphi^{\ox n}$ as the input state in the product strategy, we have
\begin{align}
\ST_{\pro, \inf}(\ve|\cN\|\cM) \geq \liminf_{n\to \infty} \frac1n  d_{H}^\ve([\cN(\varphi)]^{\ox n}\|[\cM(\varphi)]^{\ox n}) = D(\cN(\varphi)\|\cM(\varphi)) = d(\cN\|\cM),
\end{align}
where the first equality follows from the quantum Stein's lemma~\cite{hiai1991proper,ogawa2000strong} and the second equality follows from the optimality assumption of $\varphi$. For the converse part, consider any input states $\Motimes_{i=1}^n \varphi_i$ with $\varphi_i \in \md(A_i)$ and $\alpha > 1$ we have 
\begin{align}
\frac1n D_H^\ve\left(\Motimes_{i=1}^n \cN(\varphi_i)\Big\|\Motimes_{i=1}^n \cM(\varphi_i)\right) & \leq \frac1n \tD_\alpha\left(\Motimes_{i=1}^n \cN(\varphi_i)\Big\|\Motimes_{i=1}^n \cM(\varphi_i)\right) + \frac1n  \frac{\alpha}{\alpha-1}\log\frac{1}{1-\varepsilon}\\
& = \frac1n  \sum_{i=1}^n \tD_\alpha\left(\cN(\varphi_i)\big\|\cM(\varphi_i)\right)+ \frac1n  \frac{\alpha}{\alpha-1}\log\frac{1}{1-\varepsilon}\\
& \leq \td_\alpha(\cN\|\cM)+ \frac1n  \frac{\alpha}{\alpha-1}\log\frac{1}{1-\varepsilon},
\end{align}
where the first inequality follows from Eq.~\eqref{eq: DH and sandwiched}, the first equality follows from the additivity of sandwiched \Renyi divergence under tensor product states, the second inequality follows from the definition of channel divergence. Taking the supremum of all input states $\Motimes_{i=1}^n \varphi_i$ and taking the limit of $n \to \infty$, we have
\begin{align}
  \ST_{\pro, \sup}(\ve|\cN\|\cM) = \limsup_{n\to \infty} \frac{1}{n} d_{H,\pro}^\ve(\cN^{\ox n}\|\cM^{\ox n}) \leq \td_\alpha(\cN\|\cM).
\end{align}
Finally, taking $\alpha \to 1$ and applying Lemma~\ref{lem: one shot continuity} we have the converse part.
\end{proof} \vspace{0.2cm}

\vspace{0.2cm}

Note that we can actually extend the input choices of product strategy to convex combination of tensor product states $\sum_{j=1}^m p_j (\Motimes_{i=1}^n \varphi_{i,j})$. In this case Theorem~\ref{thm: stein exponent prodcut} still holds by adding an extra step in the proof of the converse part and using the joint quasi-convexity of the sandwiched \Renyi divergence (e.g.~\cite[Corollary 3.16]{mosonyi2015quantum}). This indicates that shared randomness between the input states for each uses of the channel will not help to get a faster convergence rate of the Type-II error for channel discrimination.

\begin{theorem}[Coherent strategy]\label{thm: stein exponent coherent}
Let $\mN,\mM\in\cptp(A\to B)$ be two quantum channels and $\varepsilon\in(0,1)$ be a fixed error. If these channels exhibit the strong converse property as defined in Definition~\ref{def: strong converse property}, then it implies that
\begin{align}
\ST_{\coh, \sup}(\ve|\cN\|\cM) = \ST_{\coh, \inf}(\ve|\cN\|\cM)  = d^\reg(\cN\|\cM).
\end{align}
\end{theorem}

\begin{proof}
The assertion is a combination of Eq.~\eqref{eq: dH lower bound} (achievability) and a restatement of Theorem~\ref{thm: channel stein lemma} (converse) by noting that
\begin{align}
    d_{H,\coh}^{\varepsilon}(\mN^{\otimes n}\big\|\mM^{\otimes n}) = d_{H}^{\varepsilon}(\mN^{\otimes n}\big\|\mM^{\otimes n}),
\end{align}
where the l.h.s. is the operational definition and the r.h.s. is the mathematical definition.
\end{proof} \vspace{0.2cm}

\begin{theorem}[Sequential strategy]\label{thm: stein exponent sequential}
  Let $\cN,\cM \in \cptp(A\to B)$ be two quantum channels. Let $\cI \in \cptp(A \to A)$ be the identity channel. Then if the exponentially strong converse property, as defined in Definition~\ref{def: exponentially strong converse property}, holds for the channels $\cI \ox \cN$ and $\cI \ox \cM$, this implies that
\begin{align}
\ST_{\seq,\sup}(\ve|\cN\|\cM) = \ST_{\seq,\inf}(\ve|\cN\|\cM) =  D^A(\cN\|\cM).
\end{align}
\end{theorem}

\begin{proof}
By definition it is clear that $\ST_{\seq,\inf}(\ve|\cN\|\cM)$ is monotone increasing in $\ve$. Thus for any fixed $\ve \in (0,1)$ we have 
\begin{align}\label{eq: seq tmp1}
\ST_{\seq,\inf}(\ve|\cN\|\cM) \geq \lim_{\ve \to 0} \ST_{\seq,\inf}(\ve|\cN\|\cM) = D^A(\cN\|\cM) ,   
\end{align} 
where the equality follows from~\cite[Theorem 6]{WW2019}. 

Next we prove the converse part. For any $\psi_n\in\md(RA^n)$, $\cP_i \in  \cptp(R_iB_i \to R_{i+1}A_{i+1})$, denote 
\begin{align}
\rho_n & = \cN\circ \cP_{n-1} \circ \cdots \circ \cP_2 \circ \cN \circ \cP_1 \circ \cN (\psi_n)\\
\sigma_n & = \cM\circ \cP_{n-1} \circ \cdots \circ \cP_2 \circ \cM \circ \cP_1 \circ \cM (\psi_n).
\end{align}
Due to Eq.~\eqref{eq: DH and sandwiched}, it holds for any $\alpha > 1$ that
\begin{align}
\frac{1}{n}D_H^{\varepsilon}(\rho_n\|\sigma_n)\leq \frac{1}{n}\tD_\alpha(\rho_n\|\sigma_n)+\frac{1}{n}\frac{\alpha}{\alpha-1}\log\frac{1}{1-\varepsilon}.
\end{align}
Note that for any quantum state $\rho,\sigma$ and quantum channels $\cE,\cF$, we have by definition 
\begin{align}
\tD_\alpha(\cE(\rho)\|\cF(\sigma)) \leq \tD_\alpha^A(\cE\|\cF) + \tD_\alpha(\rho\|\sigma).  
\end{align} 
By using this relation and the data-processing inequality of $\tD_\alpha$ iteratively, we have $\tD_\alpha(\rho_n\|\sigma_n) \leq n \tD_\alpha^A(\cN\|\cM)$.
This gives
\begin{align}
\frac{1}{n}D_H^{\varepsilon}(\rho_n\|\sigma_n)\leq \tD_\alpha^A(\cN\|\cM)+\frac{1}{n}\frac{\alpha}{\alpha-1}\log\frac{1}{1-\varepsilon}.
\end{align}
Taking on both sides the supremum over all sequential strategies following by the limit $n\to \infty$ gives
\begin{align}
\ST_{\seq,\sup}(\ve|\cN\|\cM) =\limsup_{n\to\infty}\frac{1}{n}D_{H,\seq}^{\varepsilon}\left(\mN^{\otimes n}\big\|\mM^{\otimes n}\right)\leq \tD^{A}_\alpha(\mN\|\mM)\;.
\end{align}
Since the above inequality holds for all $\alpha>1$, by taking $\alpha\to 1^+$ and using Corollary~\ref{thm: sandwiched channel continuity amortized} we have
\begin{align}\label{eq: seq tmp2}
   \ST_{\seq,\sup}(\ve|\cN\|\cM)\leq D^{A}(\mN\|\mM).
\end{align}
Combining Eqs.~\eqref{eq: seq tmp1} and~\eqref{eq: seq tmp2}, we have the complete proof.
\end{proof} \vspace{0.2cm}

\vspace{0.2cm}

Note that Theorem~\ref{thm: stein exponent coherent} and~\ref{thm: stein exponent sequential} have been proved in~\cite[Theorem 3 and Theorem 6]{WW2019} for vanishing $\ve$. But the above results are stronger as they hold for any fixed $\ve \in (0,1)$ without the need to take ${\ve \to 0}$.

\subsection{Strong converse exponent}

In the task of state discrimination, the strong converse exponent is defined by
\begin{align}
\SC(r|\rho\|\sigma):=\inf_{\{\Pi_n\}}\left\{-\liminf_{n\to +\infty} \frac{1}{n} \log \tr \rho^{\ox n} \Pi_n: \limsup_{n\to +\infty} \frac{1}{n} \log \tr \sigma^{\ox n}\Pi_n \leq -r\right\},
\end{align}
where the infimum is taken over all possible sequences of quantum tests $\{\Pi_n\}_{n\in \mathbb{N}}$ satisfying the condition. It has been shown in~\cite[Theorem 4.10]{mosonyi2015quantum} that this exponent is precisely characterized by:
\begin{align}\label{eq: strong converse exponenet state version}
\SC(r|\rho\|\sigma) = \sup_{\alpha > 1} \frac{\alpha-1}{\alpha} \left[r - \widetilde D_{\alpha}(\rho\|\sigma)\right].
\end{align}
We aim to extend this result to the channel case. 

Let us start by defining the strong converse exponent of channel discrimination.

\begin{definition}[Strong converse exponent] Let $\cN,\cM \in \cptp(A\to B)$ and $r>0$. The strong converse exponents of channel discrimination by the strategy class $\Omega \in \{\pro,\coh,\seq\}$ without quantum memory assistance are defined by
\begin{align}\label{eq: strong converse exponent definition}
\SC_{\Omega}(r|\mN\|\mM):= \inf_{(S_n,\Pi_n) \in \Omega} \left\{- \liminf_{n\to +\infty}\frac{1}{n} \log (1-\alpha_n(S_n,\Pi_n)): \limsup_{n\to +\infty} \frac{1}{n} \log \beta_n(S_n,\Pi_n) \leq -r \right\},
\end{align}
where the infimum is taken over all possible strategies $(S_n,\Pi_n) \in \Omega$ satisfying the condition.
\end{definition}

\begin{theorem}[Product strategy]\label{thm: strong converse exponent product}
Let $\cN,\cM \in \cptp(A\to B)$ and $r>0$. Then it holds that
  \begin{align}
\SC_{\pro}(r|\mN\|\mM) = \sup_{\alpha > 1} \frac{\alpha-1}{\alpha} \left[r - \widetilde d_{\alpha}(\mN\|\mM)\right].
  \end{align}
\end{theorem}

\begin{proof}
We first prove the converse part which closely follows the proof of its state analog in~\cite[Lemma 4.7]{mosonyi2015quantum}. For any product strategy $(\{\varphi_i\}_{i=1}^n, \Pi_n)$ with input states  $\varphi_i \in \md(A_i)$ and measurement opeartor $0 \leq \Pi_n \leq I$. Let $\rho_n := \mN^{\otimes n}(\Motimes_{i=1}^n \varphi_i)$, $\sigma_n := \mM^{\otimes n}(\Motimes_{i=1}^n \varphi_i)$ be the output states and $p_n:=(\tr \rho_n \Pi_n, \tr \rho_n (I-\Pi_n))$ and $q_n:= (\tr \sigma_n \Pi_n, \tr \sigma_n (I-\Pi_n))$ be the post-measurement states. Then the Type-I error is $\alpha_n = \tr[(I-\Pi_n) \rho_n]$ and the Type-II error is $\beta_n = \tr[\Pi_n \sigma_n]$.  By definition it suffices to consider sequences $(\{\varphi_i\}_{i=1}^n, \Pi_n)$ such that $\limsup_{n\to +\infty} \frac{1}{n} \log \beta_n \leq -r$.  From the data-processing of the sandwiched \Renyi divergence, we have for any $\alpha > 1$ that 
\begin{align}
  \widetilde D_{\alpha}(\rho_n\|\sigma_n)  & \geq \widetilde D_{\alpha}(p_n\|q_n)\notag\\
  & \geq \frac{1}{\alpha - 1} \log \left[(\tr \rho_n \Pi_n)^{\alpha} (\tr \sigma_n \Pi_n)^{1-\alpha}\right]= \frac{\alpha}{\alpha - 1} \log (1-\alpha_n) - \log \beta_n.
\end{align}
This can be equivalently written as
\begin{align}
  -\frac{1}{n} \log(1-\alpha_n) \geq \frac{\alpha-1}{\alpha} \left[- \frac{1}{n} \log \beta_n - \frac{1}{n} \widetilde D_{\alpha}(\rho_n\|\sigma_n)\right].
\end{align}
By the assumption of $(\{\varphi_i\}_{i=1}^n, \Pi_n)$ and taking $\limsup_{n\to \infty}$ on both sides, we have 
\begin{align}
  - \liminf_{n\to +\infty} \frac{1}{n} \log(1-\alpha_n) \geq \frac{\alpha-1}{\alpha} \left[r - \liminf_{n\to +\infty}\frac{1}{n} \widetilde D_{\alpha}(\rho_n\|\sigma_n)\right].
\end{align}
By the additivity of sandwiched \Renyi divergence under tensor product states and the definition of channel divergence, we have $\tD_{\alpha}(\rho_n\|\sigma_n) = \sum_{i=1}^n \tD_{\alpha}(\cN(\varphi_i)\|\cM(\varphi_i))\leq n \td_\alpha(\cN\|\cM)$. Thus
\begin{align}\label{eq: product strong converse exponent proof tmp1}
  - \liminf_{n\to +\infty} \frac{1}{n} \log(1-\alpha_n) \geq \frac{\alpha-1}{\alpha} \left[r - \td_{\alpha}(\cN\|\cM)\right].
\end{align}
Finally taking the infimum over all product strategies and the supremum over all $\alpha > 1$ on both sides, we can conclude the converse part
\begin{align}\label{eq: sc tmp1}
  \SC_{\pro}\left(r|\mN \|\mM\right) \geq \sup_{\alpha > 1} \frac{\alpha-1}{\alpha} \left[r - \widetilde d_{\alpha}(\mN\|\mM)\right].
\end{align}

We then proceed to show the achievable part. Let $\varphi \in \md(A)$ be an optimal quantum state such that $\widetilde d_{\alpha}(\mN\|\mM) = \widetilde D_{\alpha}(\mN(\varphi)\|\mM(\varphi))$. Consider the task of distinguishing quantum states $\mN(\varphi)$ and $\mM(\varphi)$. Suppose the optimal test in $\SC(r|\mN(\varphi)\|\mM(\varphi))$ is given by the sequence $\{\Pi_{n}\}_{n\in \mathbb{N}}$. Then by the quantum converse Hoeffiding theorem (see~\eqref{eq: strong converse exponenet state version}) we have
\begin{align}
   \limsup_{n\to +\infty} \frac{1}{n} \log \tr [\mM(\varphi)]^{\ox n}\Pi_{n}  & \leq -r \qquad \text{and}\\
  - \liminf_{n\to +\infty} \frac{1}{n} \log \tr [\mN(\varphi)]^{\ox n} \Pi_{n} &  = \sup_{\alpha > 1} \frac{\alpha-1}{\alpha} \left[r - \widetilde D_{\alpha}(\cN(\varphi)\|\cM(\varphi))\right].\label{eq: product strong converse exponent tmp1}
\end{align}
Note that $(\{\varphi\}_{i=1}^n,\Pi_{n})$ is a product stategy for the task of channel discrimination between $\mN^{\ox n}$ and $\mM^{\ox n}$. We have
\begin{align}
\SC_{\pro}\left(r|\mN \|\mM\right) & \leq - \liminf_{n\to +\infty} \frac{1}{n} \log \tr \mN^{\ox n}(\varphi^{\ox n})\Pi_{n}\\
  & = \sup_{\alpha > 1} \frac{\alpha-1}{\alpha} \left[r - \widetilde D_{\alpha}(\mN(\varphi)\|\mM(\varphi))\right]\\
  & = \sup_{\alpha > 1} \frac{\alpha-1}{\alpha} \left[r - \widetilde d_{\alpha}(\mN\|\mM)\right]\label{eq: sc tmp2}
\end{align}
where the first equality follows from~\eqref{eq: product strong converse exponent tmp1}, the second equality follows from the optimality assumption of $\varphi$. Combining Eqs.~\eqref{eq: sc tmp1} and~\eqref{eq: sc tmp2}, we have the complete proof.
\end{proof} \vspace{0.2cm}

\vspace{0.2cm}
Note here that one can extend the input choices of product strategy to convex combination of tensor product states $\sum_{j=1}^m p_j (\Motimes_{i=1}^n \varphi_{i,j})$. In this case Theorem~\ref{thm: strong converse exponent product} still holds by adding an additional step in the proof of the converse part and using the joint quasi-convexity of the sandwiched \Renyi divergence (e.g.~\cite[Corollary 3.16]{mosonyi2015quantum}). This indicates that shared randomness between the input states for each use of the channel will provide no advantage in reducing the convergence rate of the Type-I error.

\begin{remark}
The strong converse exponents under coherent and sequential strategies were established in~\cite[Theorem 5.5]{fawzi2021defining}. However, regarding the exact threshold for exponential convergence, their result only identifies the threshold as $\inf_{\alpha > 1} \tD^\reg_\alpha(\cN \| \cM)$. The continuity result in Theorem~\ref{thm: sandwiched channel continuity regularized} could fully determine this threshold as $D^\reg(\cN \| \cM)$ if the strong converse property can be proven.
\end{remark}

\subsection{Error exponent}

In the task of state discrimination, the error exponent is defined by
\begin{align}
\ER(r|\rho\|\sigma):=\sup_{\{\Pi_n\}}\left\{-\limsup_{n\to +\infty} \frac{1}{n} \log \tr [(I-\Pi_n)\rho^{\ox n} ]: \limsup_{n\to +\infty} \frac{1}{n} \log \tr [\Pi_n\sigma^{\ox n}] \leq -r\right\},
\end{align}
where the supremum is taken over all possible sequences of quantum tests $\{\Pi_n\}_{n\in \mathbb{N}}$ satisfying the condition. It has been shown in~\cite{hayashi2007error,nagaoka2006converse,audenaert2008asymptotic} that the error exponent is precisely given by:
\begin{align}\label{eq: error exponenet state version}
\ER(r|\rho\|\sigma) = \sup_{0<\alpha <1} \frac{\alpha-1}{\alpha} \left[r - \widebar D_\alpha(\rho\|\sigma)\right].
\end{align}
We aim to extend this result to the channel case.

\begin{definition}[Error exponent]Let $\cN,\cM \in \cptp(A\to B)$ and $r > 0$. The error exponents of quantum channel discrimination by the strategy class $\Omega \in \{\pro,\coh,\seq\}$ without quantum memory assistance are defined by
\begin{align}\label{eq: error exponent definition}
\ER_{\Omega}(r|\mN\|\mM):= \sup_{(S_n,\Pi_n) \in \Omega} \left\{- \limsup_{n\to +\infty}\frac{1}{n} \log \alpha_n(S_n,\Pi_n): \limsup_{n\to +\infty} \frac{1}{n} \log \beta_n(S_n,\Pi_n) \leq -r \right\},
\end{align}
where the supremum is taken over all possible  strategies $(S_n,\Pi_n) \in \Omega$ satisfying the condition.
\end{definition}

\begin{theorem}[Product strategy]\label{thm: error exponent product}
Let $\mN, \mM \in \cptp(A\to B)$ and $r > 0$. Then it holds that
  \begin{align}
\ER_{\pro}(r|\mN\|\mM) \geq \sup_{0<\alpha <1} \frac{\alpha-1}{\alpha} \left[r - \widebar d_\alpha(\cN\|\cM)\right].
  \end{align}
\end{theorem}

\begin{proof}
Let $\varphi \in \md(A)$ an optimal input state such that $\bd_{\alpha}(\mN\|\mM) = \bD_{\alpha}(\mN(\varphi)\|\mM(\varphi))$.
Consider the task of distinguishing quantum states $\mN(\varphi)$ and $\mM(\varphi)$. Suppose the optimal test in $\ER(r|\mN(\varphi)\|\mM(\varphi))$ is given by the sequence $\{\Pi_{n}\}_{n\in \mathbb{N}}$. Then by the quantum Hoeffding theorem~\eqref{eq: error exponenet state version} we have
\begin{align}
   \limsup_{n\to +\infty} \frac{1}{n} \log \tr [\mM(\varphi)]^{\ox n}\Pi_{n}  & \leq -r \qquad \text{and}\\
  - \limsup_{n\to +\infty} \frac{1}{n} \log (1-\tr [\mN(\varphi)]^{\ox n} \Pi_{n}) &  = \sup_{0<\alpha <1} \frac{\alpha-1}{\alpha} \left[r - \widebar D_\alpha(\cN(\varphi)\|\cM(\varphi))\right].\label{eq: product error exponent tmp1}
\end{align}
Note that $(\{\varphi\}_{i=1}^n,\Pi_n)$ is a product strategy for the task of channel discrimination between $\mN^{\ox n}$ and $\mM^{\ox n}$. Then we have
\begin{align}
\ER_{\pro}\left(r|\mN\|\mM\right) & \geq - \limsup_{n\to +\infty} \frac{1}{n} \log (1-\tr \mN^{\ox n}(\varphi^{\ox n})\Pi_{n})\\
  & = - \limsup_{n\to +\infty} \frac{1}{n} \log (1-\tr [\mN(\varphi)]^{\ox n} \Pi_{n})\\
  & = \sup_{0<\alpha < 1} \frac{\alpha-1}{\alpha} \left[r - \bD_{\alpha}(\mN(\varphi)\|\mM(\varphi))\right]\\
  & = \sup_{0<\alpha < 1} \frac{\alpha-1}{\alpha} \left[r - \bd_{\alpha}(\mN\|\mM)\right]
\end{align}
where the second equality follows from~\eqref{eq: product error exponent tmp1}, the third equality follows from the optimality assumption of $\varphi$. This completes the proof.
\end{proof} \vspace{0.2cm}

\begin{theorem}[Coherent strategy]\label{thm: error exponent coherent}
Let $\mN, \mM \in \cptp(A\to B)$ and $r > 0$. Then it holds that
  \begin{align}
\ER_{\coh}(r|\mN\|\mM) \geq \sup_{0<\alpha <1} \frac{\alpha-1}{\alpha} \left[r - \widebar d_\alpha^\reg(\cN\|\cM)\right].
  \end{align}
\end{theorem}

\begin{proof}
For any given $m \in \mathbb{N}$, let $\psi_m \in \md(A^m)$ an optimal input state such that $\bd_{\alpha}(\mN^{\ox m}\|\mM^{\ox m}) = \bD_{\alpha}(\mN^{\ox m}(\psi_m)\|\mM^{\ox m}(\psi_m))$. Denote $\rho_m := \mN^{\ox m}(\psi_m)$ and $\sigma_m := \mM^{\ox m}(\psi_m)$.
Consider the task of distinguishing quantum states $\rho_m$ and $\sigma_m$. Suppose the optimal test in $\ER(r|\rho_m\|\sigma_m)$ is given by the sequence $\{\Pi_{m,n}\}_{n\in \mathbb{N}}$. Then by the quantum Hoeffding theorem (see~\eqref{eq: error exponenet state version}) we have
\begin{align}
   \limsup_{n\to +\infty} \frac{1}{n} \log \tr [\sigma_m]^{\ox n}\Pi_{m,n}  & \leq -r \qquad \text{and}\\
  - \limsup_{n\to +\infty} \frac{1}{n} \log (1-\tr [\rho_m]^{\ox n} \Pi_{m,n}) 
  & = \sup_{0<\alpha <1} \frac{\alpha-1}{\alpha} \left[r - \widebar D_\alpha(\rho_m\|\sigma_m)\right].\label{eq: error exponent tmp1}
\end{align}
Note that $(\psi_m^{\ox n},\Pi_{m,n})$ is a coherent strategy for the task of channel discrimination between $\mN^{\ox mn}$ and $\mM^{\ox mn}$, satisfying
\begin{align}
  \limsup_{n\to +\infty} \frac{1}{mn} \log \tr \mM^{\ox mn}(\psi_m^{\ox n}) \Pi_{m,n}  & \leq - \frac{r}{m}.
 \end{align} 
 Then we have
\begin{align}
\ER_{\coh}\left(\frac{r}{m}\Big|\mN\Big\|\mM\right) & \geq - \limsup_{n\to +\infty} \frac{1}{mn} \log (1-\tr \mN^{\ox mn}(\psi_m^{\ox n})\Pi_{m,n})\\
  & = - \limsup_{n\to +\infty} \frac{1}{mn} \log (1-\tr [\rho_m]^{\ox n} \Pi_{m,n})\\
  & = \frac{1}{m}  \sup_{0<\alpha < 1} \frac{\alpha-1}{\alpha} \left[r - \bD_{\alpha}(\rho_m\|\sigma_m)\right]\\
  & = \frac{1}{m}  \sup_{0<\alpha < 1} \frac{\alpha-1}{\alpha} \left[r - \bd_{\alpha}(\mN^{\ox m}\|\mM^{\ox m})\right]\\
  & =  \sup_{0<\alpha < 1} \frac{\alpha-1}{\alpha} \left[\frac{r}{m} - \frac{1}{m}\bd_{\alpha}(\mN^{\ox m}\|\mM^{\ox m})\right],
\end{align}
where the second equality follows from~\eqref{eq: error exponent tmp1}, the third equality follows from the optimality assumption of $\psi_m$. Replacing $r/m$ as $r$, we have
\begin{align}\label{eq: error exponent proof tmp3}
\ER_{\coh}\left(r|\mN\|\mM\right) \geq \sup_{0<\alpha < 1} \frac{\alpha-1}{\alpha} \left[r - \frac{1}{m}\bd_{\alpha}(\mN^{\ox m}\|\mM^{\ox m})\right].
\end{align}
Since~\eqref{eq: error exponent proof tmp3} holds for any integer $m \in \mathbb{N}$, we have
\begin{align}
\ER_{\coh}\left(r|\mN\|\mM\right) & \geq \sup_{m\in \mathbb{N}}\sup_{0<\alpha < 1} \frac{\alpha-1}{\alpha} \left[r - \frac{1}{m}\bd_{\alpha}(\mN^{\ox m}\|\mM^{\ox m})\right]\\
   & = \sup_{0<\alpha < 1} \sup_{m\in \mathbb{N}} \frac{\alpha-1}{\alpha} \left[r - \frac{1}{m}\bd_{\alpha}(\mN^{\ox m}\|\mM^{\ox m})\right]\\
   & = \sup_{0<\alpha < 1} \frac{\alpha-1}{\alpha} \left[r - \sup_{m\in \mathbb{N}} \frac{1}{m}\bd_{\alpha}(\mN^{\ox m}\|\mM^{\ox m})\right]\\
   & = \sup_{0<\alpha < 1} \frac{\alpha-1}{\alpha} \left[r - \bd_{\alpha}^\reg(\mN\|\mM)\right].
\end{align} 
This completes the proof.
\end{proof} \vspace{0.2cm}

\section{Quantum communication as quantum channel discrimination}
\label{sec: Quantum communication as quantum channel discrimination}

Quantum communication via quantum channels forms the cornerstone of future quantum networks~\cite{fang2023quantum} and the quantum channel capacity is a central question in quantum Shannon theory~\cite{Wang2017a,Wang2017d,fang2021geometric,Fang2018}. In this section, we present a perspective by framing the study of quantum communication problems as quantum channel discrimination tasks. This perspective offers deeper insights into the intricate relationships between channel capacities, channel discrimination, and the mathematical properties of quantum channel divergences. One one hand, leveraging this connection, we demonstrate that the channel coherent information and quantum channel capacity can be precisely characterized as Stein exponent for discriminating between two quantum channels under product and coherent strategies without quantum memory assistance, respectively. Furthermore, we show that the strong converse property of quantum channel capacity can be established if the channels being discriminated exhibit the strong converse property. On the other hand, the extreme non-additivity of quantum channel capacity implies a similar fundamental property for the unstabilized channel divergence, which can be of independent interest for future studies. 

\subsection{Operational interpretation of quantum channel capacity}

In this subsection we discuss quantum channel communication and its operational interpretation in the context of quantum channel discrimination. The coding scheme for $n$ uses of the channel is depicted in Figure~\ref{fig: quantum coding framework}. We are given a quantum channel $\cN \in \cptp(A\to B)$ and denote by $\cN^{\ox n}$ the $n$-fold parallel repetition of this channel. An \emph{entanglement transmission code} for $\cN^{\ox n}$ is given by a triplet $\{|K|, \cE, \cD\}$, where $|K|$ is the local dimension of a maximally entangled state $\Phi_{EK}:= \frac{1}{|E|}\sum_{i,j=1}^{|E|} \ket{ii}\bra{jj}$ that is to be
transmitted over $\cN^{\ox n}$. The quantum channels $\cE \in \cptp(K \to A^n)$
and $\cD \in \cptp(B^n \to K)$ are encoding and
decoding operations, respectively. Denote the outcome state after the coding strategy by 
\begin{align}
  \rho_{EK}[\cE,\cD] & := \cI_E\ox \cD_{B\to K}\circ \cN_{A\to B}\circ \cE_{K\to A}(\Phi_{EK}).
\end{align}
With this in hand, we now say that a triplet $\{r, n, \ve\}$ is \emph{achievable} on the channel $\cN$ if there exists an entanglement transmission code satisfying
\begin{align}
  \frac{1}{n} \log |K| \geq r \qquad \text{and} \qquad F\left(\Phi_{EK},\rho_{EK}[\cE,\cD]\right) \geq 1-\ve,
\end{align}
where $F(\rho,\sigma) := (\|\sqrt{\rho}\sqrt{\sigma}\|_1)^2$ is the quantum fidelity and $\|\cdot\|_1$ is the trace norm. If one of the state is pure, we have the simplification $F(\ket{\psi}\bra{\psi},\sigma) = \tr [\ket{\psi}\bra{\psi}\sigma]$.

\begin{figure}[H]
\centering
\includegraphics[width=9cm]{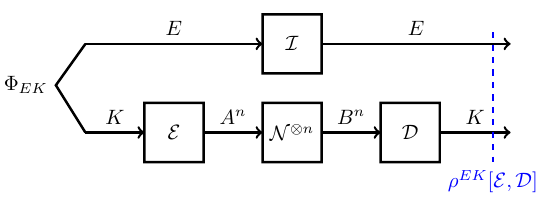}
\caption{Coding Scheme for entanglement transmission over $n$ uses of a channel $\cN \in \cptp(A\to B)$. The system $E$ and $K$ are isomorphic. The encoder $\cE \in \cptp(K \to A^n)$ encodes the part $K$ of the maximally entangled state $\Phi_{EK}$ into the channel input systems. Later, the decoder $\cD \in \cptp(B^n \to K)$ recovers the state from the channel output systems. The map $\cI \in \cptp(E \to E)$ is the identity map. The final state after the coding strategy is denoted as $\rho_{EK}[\cE,\cD]$ and the performance of the code is quantified using the fidelity $F\left(\Phi_{EK},\rho_{EK}[\cE,\cD]\right)$.}
\label{fig: quantum coding framework}
\end{figure}

When considering a single use of the channel, the one-shot quantum capacity, which establishes the boundary of all achievable triples $\{r,1,\ve\}$, is defined as follows:

\begin{definition}
Let $\cN \in \cptp(A\to B)$ be a quantum channel and $\ve \in (0,1)$ be a fixed error. The one-shot quantum capacity of $\cN$ is defined by
\begin{align}\label{eq: definition one shot capcity}
Q^{(1)}(\cN,\ve) :=  \sup_{\substack{|K|=|E| \in \mathbb{N}\\ \cE \in \cptp(K \to A) \\ \cD \in \cptp(B \to K)}} \ \Big\{ \log |K| : \tr \left( \rho_{EK}[\cE,\cD] \cdot \Phi_{EK}  \right)\geq 1-\ve\Big\}.
\end{align} 
Then the quantum capacity is defined as the asymptotic limit 
\begin{align}
  Q(\cN) := \lim_{\ve \to 0} \lim_{n\to \infty} \frac1n Q^{(1)}(\cN^{\ox n},\ve).
\end{align}
\end{definition}

The well-established Lloyd-Shor-Devetak Theorem~\cite{lloyd1997capacity,shor2002quantum,devetak2005private} states that the quantum capacity of a channel can be expressed in terms of a regularized channel coherent information:
\begin{align}\label{eq: quantum capacity theorem}
  Q(\cN) = \lim_{n\to \infty} \frac1n I_c(\cN^{\ox n}) = \sup_{n \in \mathbb{N}} \frac1n I_c(\cN^{\ox n}),
\end{align}
where the channel coherent information is defined by
\begin{align}
  I_c(\cN) := \sup_{\rho \in \md(EA)} - H(E|B)_{\rho} \quad \text{with}\quad \rho_{EB} = \cI_E\ox \cN_{A\to B}(\rho_{EA}),
\end{align}
and the supremum is taken over all density matrices $\rho$ on system $E\ox A$ and $E$ is isomorphic to $A$.

Based on the notion of unstabilized channel divergence, we can rewrite the channel coherent information and quantum capacity as follows:

\begin{theorem}
\label{lem: coherent information channel divergence}
For any quantum channel $\cN \in \cptp(A\to B)$, it holds that
\begin{align}
  I_c(\cN)  =  \ulD\left(\,\cI_E \ox \cN_{A\to B}\big\|\cR^{\pi}_E \ox \cN_{A \to B}\,\right) - \log |E| &,\label{eq: channel coherent information unstabilized}\\
  Q(\cN)  = \ulD^\reg\left(\,\cI_E \ox \cN_{A\to B}\big\|\cR^{\pi}_E \ox \cN_{A \to B}\,\right) - \log |E| &,\label{eq: channel capacity unstabilized}
\end{align}
where $E$ is isomorphic to $A$ and $\cR_E^\pi \in \cptp(E \to E)$ represents a replacer channel that maps any input state to a maximally mixed state $\pi \in \md(E)$.
\end{theorem}

\begin{proof}
For any given state $\rho \in \md(EB)$, we have
\begin{align}
- H(E|B)_\rho   
& = D(\rho_{EB}\|I_E\ox [\tr_E \rho_{EB}]) = D(\rho_{EB}\|\cR^{\pi}_E \ox \cI_B (\rho_{EB})) - \log |E|.
\end{align}
Let $\rho_{EB} = \cI_E\ox \cN_{A\to B}(\rho_{EA})$ and take supremum over all $\rho \in \md(EA)$. We get
\begin{align}
I_c(\cN) & = \sup_{\rho \in \md(EA)} D\left(\cI_E\ox \cN_{A\to B}(\rho_{EA})\big\|(\cR^{\pi}_E \ox \cI_B) \circ (\cI_E\ox \cN_{A\to B})(\rho_{EA})\right) - \log |E|\\
& = \ulD\left(\cI_E \ox \cN_{A\to B}\big\|\cR^I_{E}\ox \cN_{A\to B}\right) - \log|E|.
\end{align}
The equation~\eqref{eq: channel capacity unstabilized} directly follows from~\eqref{eq: channel coherent information unstabilized} by taking a regularization on both sides. That is,
\begin{align}
  Q(\cN) & = \lim_{n\to \infty} \frac1n I_c(\cN^{\ox n})\\
  & = \lim_{n\to \infty} \frac1n \ulD\left(\,\cI_{E^n} \ox [\cN_{A\to B}]^{\ox n}\big\|\cR_{\pi^{\ox n}} \ox [\cN_{A\to B}]^{\ox n}\,\right) - \log |E|\\
  & = \lim_{n\to \infty} \frac1n \ulD\left(\,\big[\cI_{E} \ox \cN_{A\to B}\big]^{\ox n}\big\|\big[\cR_{\pi} \ox \cN_{A\to B}\big]^{\ox n}\,\right) - \log |E|\\
  & = \ulD^\reg\left(\,\cI_E \ox \cN_{A\to B}\big\|\cR^{\pi}_E \ox \cN_{A \to B}\,\right) - \log |E|,
\end{align}
which completes the proof.
\end{proof} \vspace{0.2cm}

\begin{remark}(Operational interpretation.)
From the operational meaning of $d^\reg$, we can understand quantum capacity as the Stein exponent of channel discrimination between the ideal case $\cI_E\ox \cN_{A\to B}$ and the worst case $\cR^{\pi}_E\ox \cN_{A\to B}$. Noting that $\ulD^\reg(\cI_E\|\cR^{\pi}_E) = \log |E|$, we can also write 
\begin{align}
Q(\cN) = \ulD^\reg(\cI_E \ox \cN_{A\to B}\|\cR^{\pi}_E \ox \cN_{A \to B}) - \ulD^\reg(\cI_E\|\cR^\pi_E),    
\end{align}
indicating that the quantum capacity of a channel $\cN$ can be understood as the ``power'' of this channel as a catalyst to discriminate the perfect channel $\cI_E$ and the completely useless channel $\cR^{\pi}_E$ for quantum communication.  
\end{remark}

Drawing upon the correspondence established in Theorem~\ref{lem: coherent information channel divergence} and the extreme non-additivity of channel coherent information as shown in~\cite{cubitt2015unbounded} (where an unbounded number of channel uses may be necessary to detect quantum capacity), we can infer that the unstabilized quantum channel divergence can also exhibit extreme non-additivity.

\begin{theorem}\label{thm: extremely non-additive}
Let $d$ be the unstablized quantum channel divergence induced by the Umegaki relative entropy. Then $d$ is extremely non-additive. That is, for any $n\in \mathbb{N}$, there exists quantum channels $\cN,\cM \in \cptp(A\to B)$ such that
\begin{align}
    d(\cN^{\ox n}\|\cM^{\ox n}) < d^\reg(\cN\|\cM).
\end{align}
\end{theorem}

\begin{proof}
It has been shown in ~\cite{cubitt2015unbounded} that for any $n\in \mathbb{N}$, there exists quantum channels $\cE\in \cptp(A\to B)$ such that $I_c(\cE^{\ox n}) = 0 < Q(\cE)$. Then by the relation in Theorem~\ref{lem: coherent information channel divergence}, we can take $\cN = \cI \ox \cE$ and $\cM = \cR^{\pi}_E \ox \cE$. This gives $d(\cN^{\ox n}\|\cM^{\ox n}) = \log |A| < d^\reg(\cN\|\cM)$.
\end{proof} \vspace{0.2cm}

\subsection{One-shot converse bound for quantum channel capacity}
\label{sec: One-shot converse bound}

Here we establish a converse bound for one-shot quantum capacity, which can be seen as a smoothed analogue of channel coherent information. The one-shot converse bound is mostly inspired by the channel divergence formula of coherent information~\eqref{eq: channel capacity unstabilized}. That is, the channel coherent information as well as quantum capacity characterize the distinguishability between the channel $\cI_E \ox \cN_{A\to B}$ and the CP map $\cR^{I}_E \ox \cN_{A\to B}$ (here we use the identity operator $I$ instead of $\pi$ to obsorb the constant factor $\log |E|$). It is thus convenient for us to consider a conceptual process in Figure~\ref{fig: quantum coding framework 1} which replaces the identity map $\cI_E$ in Figure~\ref{fig: quantum coding framework} with the CP map $\cR^I_E$. Its final state is denoted as
\begin{align}\label{eq: definition of sigam}
\sigma_{EK}[\cE,\cD] & := \cR^I_E\ox \cD_{B\to K}\circ \cN_{A\to B}\circ \cE_{K\to A}(\Phi_{EK}) = \cR^I_E\ox \cI_K(\rho_{EK}[\cE,\cD]).
\end{align}

\begin{figure}[H]
\centering
\includegraphics[width=9cm]{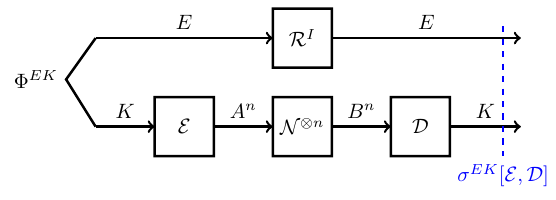}
\caption{A conceptual process that replaces $\cI_E$ in Figure~\ref{fig: quantum coding framework} with the CP map $\cR^I_E$, and gives the final state $\sigma_{EK}[\cE,\cD]$.}
\label{fig: quantum coding framework 1}
\end{figure}

\vspace{0.2cm}

\begin{theorem}[One-shot converse bound]
\label{thm: capacity one shot upper bound}
For any $\cN \in \cptp(A\to B)$ and $\ve \in (0,1)$, it holds that
\begin{align}
Q^{(1)}(\cN,\ve) \leq \sup_{|E| \in \mathbb{N}} {d_H^\ve}\left(\cI_E\ox \cN_{A\to B}\big\|\cR^I_E\ox \cN_{A\to B}\right),
\end{align}
where the supremum is taken over $E$ of arbitrary dimension. 
\end{theorem}

\begin{proof} 
For any entanglement transmission code $\{|K|,\cE,\cD\}$ such that $\tr \Phi_{EK} \rho_{EK}[\cE,\cD] \geq 1-\ve$.
We have a key observation that 
\begin{align}
\tr \Phi_{EK} \sigma_{EK}[\cE,\cD] & = \tr \Phi_{EK} \left\{\cR^I_E\ox \cI_K(\rho_{EK}[\cE,\cD])\right\}\\
& = \tr \Phi_{EK} \left\{I_E \ox \tr_E(\rho_{EK}[\cE,\cD])\right\} \\
& = \tr \left\{\tr_E(\Phi_{EK})\right\} \left\{\tr_E(\rho_{EK}[\cE,\cD])\right\}\\
& = \tr \left\{I_K/|K|\right\} \left\{\tr_E(\rho_{EK}[\cE,\cD])\right\}\\
& = 1/{|K|}, 
\end{align}
where the first line follows by~\eqref{eq: definition of sigam}, the second line follows by definition of $\cR^I_E$, the third line follows by the identity $\tr X_{AB} (I_A \ox Y_{B}) = \tr\{\tr_A X_{AB}\}\{Y_B\}$, the last line follows by the fact that $\tr_E(\rho_{EK}[\cE,\cD])$ is a normalized quantum state.
Then we have
\begin{align}
\log |K| & = - \log \tr \Phi_{EK} \sigma_{EK}[\cE,\cD]\\
& \leq D_H^\ve\left(\rho_{EK}[\cE,\cD]\big\|\sigma_{EK}[\cE,\cD]\right)\\
& \leq D_H^\ve\left((\cI_E\ox \cN_{A\to B})\circ(\cE_{K\to A}(\Phi_{EK}))\big\|(\cR^I_E\ox \cN_{A\to B})\circ(\cE_{K\to A}(\Phi_{EK}))\right)\\
& \leq \sup_{\rho\in \md(EA)}{D_H^\ve}\left(\cI_E\ox \cN_{A\to B}(\rho_{EA})\big\|\cR^I_E\ox \cN_{A\to B}(\rho_{EA})\right)
\end{align}
where the first inequality follow becasue $\Phi_{EK}$ is a particular choice of quantum test for hypothesis testing relative entropy that satisfies $\tr \Phi_{EK} \rho_{EK}[\cE,\cD] \geq 1-\ve$, the second inequality follows by the data-processing inequalty of $D_H^\ve$ under the action of $\cI_E \ox \cD_{B\to K}$, 
the third inequality follows by relaxing $\cE_{K\to A}(\Phi_{EK})$ to all density matrices on system $E \ox A$. Finally taking supremum over all possible codes $\{|K|,\cE,\cD\}$, we have the desired result.
\end{proof} \vspace{0.2cm}

\begin{corollary}
\label{cor: capacity one shot upper bound Dmax}
For any $\cN \in \cptp(A\to B)$ and $\ve \in (0,1)$, $\delta \in (0,1)$ and $\ve + \delta < 1$, it holds that
\begin{align}\label{eq: Dmax converse bound}
Q^{(1)}(\cN,\ve) \leq  {\ulD_{\max}^\delta}\left(\cI_E\ox \cN_{A\to B}\big\|\cR^\pi_E\ox \cN_{A\to B}\right) + \log\frac{1}{1-\ve - \delta} - \log |E|,
\end{align}
where $E$ is isomorphic to $A$. 
\end{corollary}
\begin{proof}
Combining Theorem~\ref{thm: capacity one shot upper bound} and Eq.~\eqref{eq: Dh and Dmax 1} we have
\begin{align}
 Q^{(1)}(\cN,\ve) \leq \sup_{\substack{|E| \in \mathbb{N} \\ \rho\in \md(EA)}} D_{\max}^\delta\left(\cI_E\ox \cN_{A\to B}(\rho_{EA})\big\|\cR^I_E\ox \cN_{A\to B}(\rho_{EA})\right) + \log \frac{1}{1-\ve - \delta}.
\end{align}
Since $D_{\max}^\delta$ is jointly quasi-convex~\cite[Lemma 7]{datta2013smooth}, the optimization can be restricted to pure states. Furthermore, due to the isometry invariance property of $D_{\max}^\delta$, we can, without loss of generality, assume that $E$ is isomorphic to $A$. Finally, noting that $\cR_E^I(\cdot) = |E|\cR_E^\pi(\cdot)$ and $D_{\max}^\delta(\rho\|a\sigma) = D_{\max}^\delta(\rho\|\sigma) - \log a$, we have the asserted result.
\end{proof} \vspace{0.2cm}

\subsection{Towards the strong converse property of quantum channel capacity}

Consider any entanglement transmission code with an achievable triplet $\{r,n,\ve\}$. The strong converse property of channel $\cN$ is that if the communication rate $r > Q(\cN)$, then the communication error $\ve$ converges to one as $n$ goes to infinity. Similar to the proof of Theorem~\ref{thm: channel stein lemma}, this can be equivalently expressed by
\begin{align}
    \limsup_{n\to \infty} \frac{1}{n} Q^{(1)}(\cN^{\ox n},\ve) \leq Q(\cN), \quad \forall \ve \in (0, 1).
\end{align}
The strong converse property of quantum capacity is a long-standing open problem in quantum information theory. Upon the connnection between quantum communication and channel discrimination, we show that strong converser property for channel discrimination implies the strong converse property of quantum capacity.

\begin{theorem}[Strong converse property]\label{thm: strong converse theorem}
Let $\cN \in \cptp(A\to B)$ be a quantum channel, $\cI \in \cptp(E\to E)$ be the identity channel with $E$ isomorphic to $A$ and $\cR_E^\pi \in \cptp(E\to E)$ be the replacer channel. Then if the channels $\cI_E \ox \cN_{A\to B}$ and $\cR_E^\pi \ox \cN_{A\to B}$ exhibit the strong converse property, as defined in Definition~\ref{def: strong converse property}, this implies the strong converse property of the channel capacity for $\cN$.
\end{theorem}
\begin{proof}
For any $\ve\in (0,1)$, let $\delta \in (0,1)$ such that $\ve + \delta < 1$. Then  
\begin{align}
\limsup_{n\to \infty } \frac1n Q^{(1)}(\cN^{\ox n},\ve) 
& \leq  \limsup_{n\to \infty } \frac1n {\ulD_{\max}^\delta}\left((\cI_E\ox \cN_{A\to B})^{\ox n}\big\|(\cR^\pi_E\ox \cN_{A\to B})^{\ox n}\right) - \log |E|\\
& \leq d^\reg\left(\cI_E\ox \cN_{A\to B}\big\|\cR^\pi_E\ox \cN_{A\to B}\right)) - \log |E|\\
& = Q(\cN),
\end{align}
where the first inequality follows from Corollary~\ref{cor: capacity one shot upper bound Dmax}, the second inequality follows from Theorem~\ref{thm: strong Dmax AEP} and the equality follows from Theorem~\ref{lem: coherent information channel divergence}. This completes the proof.
\end{proof} \vspace{0.2cm}

\section{Discussion}

In conclusion, this work advances the understanding of the ultimate limits of quantum channel discrimination and quantum communication by developing versatile tools and frameworks rooted in unstabilized channel divergence. We address key open problems, such as improving bounds on hypothesis testing relative entropy, proving additivity for channel divergences, and establishing a quantum channel analog of Stein's lemma. Our unified approach links channel discrimination strategies with operational regimes and mathematical divergences, providing a comprehensive perspective on quantum channel discrimination across various settings. Furthermore, by framing quantum communication problems as quantum channel discrimination tasks, we uncover connections between channel capacities, channel discrimination, and operational exponents. These results bridge two core areas of quantum information theory and offer new insights for future exploration.

An initial attempt to prove the exponentially strong converse property for two general channels was presented in the first arXiv submission of this work (arXiv:2110.14842v1). However, this effort triggers the discovery of a gap in a technical lemma from~\cite{brandao2010generalization}, which undermines the validity of the original proof and leaves the problem unresolved. Notably, this gap has drawn great attention in the quantum information community since then and the generalized quantum Stein's lemma, originally proposed in~\cite{brandao2010generalization}, has been recently reproved in~\cite{hayashi2024generalized,lami2024solution}. Given our findings in this work that the strong converse property is a pivotal element for achieving a complete understanding of quantum channel discrimination, we encourage interested readers to give further investigations into this important problem. Several results from our initial analysis remain valid and could hold independent interest. These details are included in~\ref{sec: Some useful results}, and we hope they will inspire and support future efforts to resolve this challenging issue.

\paragraph{Acknowledgements.} {We thank Marco Tomamichel for pointing out the gap in an initial attempt to prove the exponentially strong converse property in our first arXiv submission of this work (arXiv:2110.14842v1). K.F. is supported by the National Natural Science Foundation of China (Grant No. 92470113 and 12404569), the Shenzhen Science and Technology Program (Grant No. JCYJ20240813113519025), the Shenzhen Fundamental Research Program (Grant No. JCYJ20241202124023031), the 1+1+1 CUHK-CUHK(SZ)-GDST Joint Collaboration Fund (Grant No. GRDP2025-022), and the University Development Fund (Grant No. UDF01003565). G.G. is supported by the Israel Science Foundation under Grant No. 1192/24. XW was partially supported by the National Key R\&D Program of China (Grant No.~2024YFE0102500),  the National Natural Science Foundation of China (Grant No.~12447107), the Guangdong Natural Science Foundation (Grant No.~2025A1515012834), the Guangdong Provincial Quantum Science Strategic Initiative (Grant No.~GDZX2403008, GDZX2403001), the Guangdong Provincial Key Lab of Integrated Communication, Sensing and Computation for Ubiquitous Internet of Things (Grant No.~2023B1212010007), the Quantum Science Center of Guangdong-Hong Kong-Macao Greater Bay Area, and the Education Bureau of Guangzhou Municipality. }

\bibliographystyle{alpha_abbrv}
\bibliography{Stein}

\appendix

\section{Attempt to solve the strong converse property}
\label{sec: Some useful results}

An initial attempt to prove the exponentially strong converse property for two general channels, $\cI \otimes \cN$ and $\cI \otimes \cM$, was presented in the first arXiv submission (arXiv:2110.14842v1). However, a flaw has been identified in the original proof, rendering it invalid and leaving the problem unresolved. Nevertheless, several preliminary results from our initial analysis remain valid and may hold independent interest. These details are included in this appendix, and we hope they will contribute to resolving the problem in future studies.

The first lemma is an analog of a result of Ogawa and Nagaoka~\cite{ogawa2000strong} that was originally used to establish the strong converse of quantum Stein's lemma. A similar result was proved by Brand\~ao and Plenio for tensor product states~\cite{brandao2010generalization}. Here we extend it further to permutation-symmetric states. 

\begin{lemma}\label{lem:ubd}
Let $\mu\in\mbb{R}$ and $\rho_n,\sigma_n\in\md(A^n)$ be symmetric under permutations of the $n$ subsystems such that 
$\supp(\rho_n)\subseteq\supp(\sigma_n)$.
Then, for any $r\in\mbb{R}$ and $s\in[0,1]$ the following relation holds
\be\label{10139}
\tr\left(\rho_{n}-2^{\mu n}\sigma_n\right)_+\leq 2^{-nrs+\log\tr[\rho^{1+s}_n]}+2^{-ns(\mu-r)+s|A|\log(1+n)+\log\tr[\rho_n\sigma^{-s}_n]}\;.
\ee
\end{lemma}

\begin{proof}
Let $\Pi$  be the projection to the positive part of $\rho_n-2^{\mu n}\sigma_n$ and $\Pi=\sum_{x=1}^{|A|^n}a_x\Pi_x$ be a decomposition of $\Pi$ into orthogonal rank-one projectors, where $a_x\in\{0,1\}$ and $\sum_x\Pi_x=I_{A^n}$ (i.e. the set $\{\Pi_x\}$ forms a von-Neumann rank-one projective measurement). 
In general, this decomposition of $\Pi$ is not unique, and the precise choice of $\{\Pi_x\}$ will be determined later on in the proof. Finally, denote by $p_x\eqdef\tr\left[\rho_n \Pi_x\right]$, $q_x\eqdef \tr\left[\sigma_n \Pi_x\right]$ (note that $p_x$ and $q_x$ depends on $n$), and let $\mi$ be the set of all $x$ for which $p_x>2^{\mu n}q_x$. Since $\left(\rho_n-2^{n\mu}\sigma_n\right)_+ = \Pi \left(\rho_n-2^{n\mu}\sigma_n\right)\Pi$, we have
\ba
\tr\left(\rho_n-2^{n\mu}\sigma_n\right)_+ =  \sum_{x}a_x\left(p_x-2^{\mu n}q_x\right) \leq \sum_{x\in\mi}\left(p_x-2^{\mu n}q_x\right)\leq \sum_{x\in\mi}p_x=\pr(\mi)\;,
\ea
where $\pr(\mi)$ is the probability of the set $\mi$ with respect to the probability distribution $\{p_x\}$. Note that the set $\mi$ can be written as
\be
\mi=\left\{x\;:\;\frac1n\log p_x>\mu+\frac1n\log q_x\right\}\;.
\ee
We would like to replace the set $\mi$ with two sets: one depends solely on $p_x$, and the other only on $q_x$. This can be done in the following way. For any $r\in\mbb{R}$ define the two sets
\be
\mi^{(1)}\eqdef\left\{x\;:\;\frac1n\log p_x\geq r\right\}\quad\text{and}\quad
\mi^{(2)}\eqdef\left\{x\;:\;\frac1n\log q_x\leq r-\mu\right\}\;.
\ee 
Note that if $x\in\mi$ then either $x\in\mi^{(1)}$ or $x\in\mi^{(2)}$.
We therefore conclude that
\be
\tr\left(\rho_n-2^{n\mu}\sigma_n\right)_+\leq \pr\left(\mi^{(1)}\right)+\pr\left(\mi^{(2)}\right)\;.
\ee
From Cram\'er's theorem~\cite{dembo1998} it follows that
\begin{align}
&-\log\pr\left(\mi^{(1)}\right)\geq\sup_{s\in[0,1]}\left\{nsr-\log\sum_xp_x^{1+s}\right\}\label{gg1}\\
&-\log\pr\left(\mi^{(2)}\right)\geq\sup_{s\in[0,1]}\left\{ns(\mu-r)-\log\sum_xp_xq_x^{-s}\right\}\;.\label{gg2}
\end{align}
We first bound~\eqref{gg1} in terms of $\rho_n$. For this purpose,
let $\Delta\in\cptp(A^n\to A^n)$ be the completely dephasing map (also a pinching map) $\Delta(\omega)\eqdef\sum_x\Pi_x\omega\Pi_x$ defined on all $\omega\in\md(A^n)$. Then, the density matrix $\Delta\left(\rho_n\right)$ is diagonal (in the basis that the operators $\{\Pi_x\}$ project to) with components $\{p_x\}$ on its diagonal. Hence, denoting by $\pi_{A^n}:=I_{A^n}/|A|^n$ the completely mixed state in $\md(A^n)$, and by $\alpha\eqdef 1+s$, we get by direct calculation that 
\ba
-\log\sum_x p_x^{1+s}=n(\alpha-1)\log|A|-(\alpha-1)\bD_{\alpha}\left(\Delta(\rho_n)\big\|\pi_{A^n}\right)\;.
\ea
Since $\bD_{\alpha}\left(\Delta(\rho_n)\big\|\pi_{A^n}\right) = \bD_{\alpha}\left(\Delta(\rho_n)\big\|\Delta(\pi_{A^n})\right) \leq \bD_{\alpha}\left(\rho_n\big\|\pi_{A^n}\right)$, we get
\begin{align}
  -\log\sum_x p_x^{1+s}\geq n(\alpha-1)\log|A|-(\alpha-1)\bD_{\alpha}\left(\rho_n\big\|\pi_{A^n}\right)=-\log\tr\left[\rho_n^\alpha\right]=-\log\tr\left[\rho^{1+s}_n\right].
\end{align}
Together with~\eqref{gg1}, this gives the first term on the r.h.s. of~\eqref{10139}. 

For the second term, observe that
\ba
\sum_xp_xq_x^{-s}=\tr\left[\Delta\left(\rho_n\right)\left(\Delta\left(\sigma_n\right)\right)^{-s}\right]=\tr\left[\rho_n\left(\Delta\left(\sigma_n\right)\right)^{-s}\right]
\;.
\ea
We now estimate this term by utilizing the symmetry of $\rho_n$ and $\sigma_n$. Since $\rho_n$ and $\sigma_n$ are symmetric under permutations of the $n$ subsystems they can be expressed as
\be
\rho_n=\bigoplus_{\lambda\in\irr(\mS_n)} I^{B_\lambda}\otimes \rho^{C_\lambda}_{\lambda}\quad\text{and}\quad
\sigma_n=\bigoplus_{\lambda\in\irr(\mS_n)} I^{B_\lambda}\otimes \sigma^{C_\lambda}_{\lambda}
\ee
where $\lambda$ represents an irrep of the natural representation of the permutation group $\mS_n$ on $A^n$, and $\rho_\lambda,\sigma_\lambda \geq 0$. We therefore have
\be
\rho_n-2^{\mu n}\sigma_n=\bigoplus_{\lambda\in\irr(\mS_n)} I^{B_\lambda}\otimes \left(\rho^{C_\lambda}_{\lambda}-2^{\mu n}\sigma_\lambda^{C_\lambda}\right)
\ee
The condition $\supp(\rho_n)\subseteq\supp(\sigma_n)$ implies that without loss of generality we can assume that $\sigma_n>0$ (otherwise we can restrict our consideration to the subspace of $\supp(\sigma_n)$ and embed $\rho_n$ in this space). Therefore, under this assumption we have that each  $\sigma_{\lambda}>0$. Let $P_{\lambda}$ be the projector to the support of $\left(\rho_{\lambda}-2^{\mu n}\sigma_\lambda\right)_+$, and let $P_{\lambda}\eqdef\sum_{j=1}^{|C_\lambda|^n}a_{\lambda, j}P_{\lambda, j}$ be a decomposition of $P_\lambda$ into orthogonal rank-one projectors, where $a_{\lambda, j}\in\{0,1\}$ and $\sum_jP_{\lambda,j}=I^{C_\lambda}$. Moreover, for each $\lambda\in\irr(\mS_n)$ decompose $I^{B_\lambda}\eqdef\sum_{k=1}^{|B_\lambda|}|\psi_{\lambda,k}\lr \psi_{\lambda,k}|^{B_\lambda}$, where $\{|\psi_{\lambda,k}\ra\}_k$ forms an orthonormal basis of $B_\lambda$. Finally, we denote by $x\eqdef\{\lambda,j,k\}$ and take $\Pi_x\eqdef |\psi_{\lambda,k}\lr \psi_{\lambda,k}|^{B_\lambda}\otimes P_{\lambda, j}^{C_\lambda}$. With this choice of $\Pi_x$ we get that
\ba
\Delta(\sigma_n)&=\bigoplus_{\lambda\in\irr(\mS_n)}\sum_{j,k}|\psi_{\lambda,k}\lr \psi_{\lambda,k}|^{B_\lambda}\otimes P_{\lambda, j}^{C_\lambda}\sigma_\lambda^{C_\lambda}P_{\lambda, j}^{C_\lambda}\\
&=\bigoplus_{\lambda\in\irr(\mS_n)}I^{B_\lambda}\otimes \sum_jP_{\lambda, j}^{C_\lambda}\sigma_\lambda^{C_\lambda}P_{\lambda, j}^{C_\lambda}\\
&=\bigoplus_{\lambda\in\irr(\mS_n)}I^{B_\lambda}\otimes \Delta_\lambda^{C_\lambda\to C_\lambda}\left(\sigma_\lambda^{C_\lambda}\right)
\ea
where each $\Delta_\lambda(\cdot)\eqdef\sum_jP_{\lambda, j}(\cdot)P_{\lambda, j}$ is a completely dephasing map in $\cptp(C_\lambda\to C_\lambda)$. Therefore,
\ba\label{10151}
\sum_xp_xq_x^{-s}&=\tr\left[\rho_n\left(\Delta\left(\sigma_n\right)\right)^{-s}\right]=\sum_{\lambda\in\irr(\mS_n)}|B_\lambda|\tr\left[\rho_\lambda\big(\Delta_\lambda(\sigma_\lambda)\big)^{-s}\right]
\;.
\ea
From the pinching inequality, for each $\lambda\in\irr(\mS_n)$ we have
$\Delta_\lambda(\sigma_\lambda)\geq \frac1{|C_\lambda|}\sigma_\lambda$. Moreover, since the function $r\mapsto r^\alpha$ is operator anti-monotone for $\alpha\in[-1,0]$ we get that
\be
\big(\Delta_\lambda(\sigma_\lambda)\big)^{-s}\leq \left(\frac1{|C_\lambda|}\sigma_\lambda\right)^{-s}\;.
\ee
Substituting this into~\eqref{10151} gives
\be
\sum_xp_xq_x^{-s}\leq \sum_{\lambda\in\irr(\mS_n)}|C_\lambda|^{s}|B_\lambda|\tr\left[\rho_\lambda\sigma_\lambda^{-s}\right]\;.
\ee
Now, since $C_\lambda$ can be viewed as a subspace of $\sym^n(A)$, its dimension cannot exceed that of $\sym^n(A)$ which itself is bounded by $(n+1)^{|A|}$. We therefore conclude that
\ba
\sum_xp_xq_x^{-s}&\leq(n+1)^{s|A|}\sum_{\lambda\in\irr(\mS_n)}|B_\lambda|\tr\left[\rho_\lambda\sigma_\lambda^{-s}\right]=(n+1)^{s|A|}\tr\left[\rho_n\sigma_n^{-s}\right]\;.
\ea
Together with~\eqref{gg2}, this gives the second term on the r.h.s. of~\eqref{10139}.
\end{proof} \vspace{0.2cm}

\vspace{0.2cm}

The next lemma shows that the eigenvalues of the output from $n$ use of a positive definite channel $\cN > 0$ (i.e., its Choi matrix is a positive definite operator) are uniformly bounded by an exponential factor.

\begin{lemma}\label{lem: infinity norm bounded}
Let $\mN\in\cptp(A\to B)$ and $\mN>0$. Then, there exists $b\in(0,1)$ such that for any $n\in\mbb{N}$
\be\label{maxin}
\max_{\rho\in\md(R^nA^n)}\left\|\mN^{\otimes n}\left(\rho_{R^nA^n}\right)\right\|_\infty\leq b^n\;.
\ee
\end{lemma}

\begin{proof}
Since $\mN>0$ we have its Choi matrix $J_{\cN} > 0$. Then there exists $\tau\in\md(B)$ with $\|\tau\|_\infty < 1$ (e.g. the maximally mixed state) and its associated replacer channel $\cR_\tau$ such that $tJ_{\cN} > J_{\cR_\tau}$ for some $t \in (0,\infty)$. Equivalently, we have $t\cN > \cR_\tau$. Set $\eps\eqdef 1/t$ and then $\mM\eqdef(\mN-\eps\mR_\tau)/(1-\eps)>0$; in particular, $\mM\in\cptp(A\to B)$ and $\mN=(1-\eps)\mM+\eps\mR_\tau$. Observe that
\be
\mN^{\otimes n}=\sum_{k=0}^n{n\choose k}(1-\eps)^k\eps^{n-k}\mF_{n,k}
\ee
where $\mF_{n,k}\in\cptp(A^n\to B^n)$ is a uniform convex combination of ${n\choose k}$ channels all having the form $\mM^{\otimes k}\otimes\mR_{\tau}^{\otimes n-k}$ up to permutations of the $n$ channels. Now, observe that
\ba
\left\|\mM^{\otimes k}\otimes\mR_\tau^{\otimes n-k}\left(\rho_{R^nA^n}\right)\right\|_\infty
=\left\|\mM^{\otimes k}\left(\rho_{R^nA^k}\right)\otimes\tau^{\otimes n-k}\right\|_\infty
\leq \left\|\tau^{\otimes n-k}\right\|_\infty=\|\tau\|_\infty^{n-k}\;.
\ea
Note that the order that $\mN$ and $\mR_\tau$ appear in the equation above does not effect this upper bound. Therefore, since $\mF_{n,k}$ is a convex combination of such channels we conclude that also
\be
\left\|\mF_{n,k}\left(\rho_{R^nA^n}\right)\right\|_\infty\leq\|\tau\|_\infty^{n-k}\;.
\ee
Hence, for any $\rho\in\md(R^nA^n)$
\ba
\left\|\mN^{\otimes n}\left(\rho_{R^nA^n}\right)\right\|_\infty&\leq\sum_{k=0}^n{n\choose k}(1-\eps)^k\eps^{n-k}\left\|\mF_{n,k}\left(\rho^{R^nA^n}\right)\right\|_\infty\\
&\leq \sum_{k=0}^n{n\choose k}(1-\eps)^k\eps^{n-k}\|\tau\|_\infty^{n-k}\\
&=\Big(1-\eps+\|\tau\|_{\infty}\eps\Big)^n\;.
\ea
The proof is completed by taking $b\eqdef1-\eps+\|\tau\|_{\infty}\eps$ which is clearly in $(0,1)$.
\end{proof} \vspace{0.2cm}

\vspace{0.2cm}

The next lemma shows that by utilizing the permutation symmetry of tensor product channels we can restrict the optimal input states in the discrimination strategies to be symmetric states. This reduces the problem from the most general form to a particular one that can be tackled more easily. A general result is given in~\cite[Proposition II.4]{leditzky2018approaches}. Here, we give an alternative proof for the hypothesis testing relative entropy.

\begin{lemma}\label{lem:restrict to symmetric states}
Let $\cN,\cM \in \cptp(A\to B)$. For any $n\in\mbb{N}$ there exists a pure state $|\varphi\ra\in\sym^n(\tilde RA)$ such that
\be
\max_{\psi\in\md(R^nA^n)}D_H^\ve\left(\mN^{\otimes n}(\psi_{R^nA^n})\big\|\mM^{\otimes n}(\psi_{R^nA^n})\right) = D_H^\ve\left(\mN^{\otimes n}(\varphi_{\tilde R^nA^n})\big\|\mM^{\otimes n}(\varphi_{\tilde R^nA^n})\right).
\ee
\end{lemma}

\begin{proof}
First recall a variational expression of the hypothesis testing relative entropy~\cite[Eq.(2)]{buscemi2017quantum}
\begin{align}\label{eq: variational of DH}
    D_H^\ve(\rho\|\sigma) = -\log \max_{t \geq 0} \big\{(1-\ve)t - \tr(t\rho - \sigma)_+\big\}\;.
\end{align}
Therefore, we have
\ba\label{eq: DH channel variational}
{D_{H}^\eps\left(\mN^{\otimes n}\big\|\mM^{\otimes n}\right)}=-\log \max_{t\geq 0} \left\{(1-\eps)t-\tr\big(t\mN^{\otimes n}(\psi_{R^nA^n})-\mM^{\otimes n}(\psi_{R^nA^n})\big)_+\right\}\;,
\ea
for some state $\psi_{R^nA^n}\in\md(R^nA^n)$. Let
\be
\omega_{XR^nA^n}=\frac1{n!}\sum_{\pi\in\mS_n}|\pi\lr\pi|_X\otimes P_\pi\psi_{R^nA^n} P_\pi^*
\ee
where $X$ is a `flag' system of dimension $|X|=n!$. By construction, the marginal state $\omega_{R^nA^n}$ is symmetric under permutations (i.e. has support on $\sym^n(RA)$), so there exists a symmetric purification of $\omega_{R^nA^n}$ which we denote by $\varphi_{C^nR^nA^n}$, where $C\cong RA$~\cite[Lemma 4.2.2]{renner2005security}. Let $\omega_{DXR^nA^n}$ be a purification of $\omega_{XR^nA^n}$ and thus also a purification of $\omega_{R^nA^n}$. Since all purifications of a density matrix are related via isometries, there exists an isometry $V_{C^n \to DX}$ such that
\begin{align}
\omega_{DXR^nA^n} = (V_{C^n \to DX}) \varphi_{C^nR^nA^n} (V_{C^n \to DX})^\dagger.
\end{align}
Taking a partial trace of the system $D$ on both sides gives
\be
\omega_{XR^nA^n}=\mE_{C^n\to X}\left(\varphi_{C^nR^nA^n}\right)\;,
\ee
where $\mE(\cdot) = \tr_D V (\cdot) V^\dagger \in \cptp(C^n \to X)$. Let $\tR\eqdef CR$, then $|\varphi_{\tR^nA^n}\ra\in\sym^n(\tR A)$ and 
\begin{align}
& \tr\Big(t\mN^{\otimes n}\big(\varphi_{\tR^nA^n}\big)-\mM^{\otimes n}\big(\varphi_{\tR^nA^n}\big)\Big)_+ \notag\\
 & \geq \tr\Big(t\mN^{\otimes n}\big(\omega_{XR^nA^n}\big)-\mM^{\otimes n}\big(\omega_{XR^nA^n}\big)\Big)_+\\
 & = \frac1{n!}\sum_{\pi\in\mS_n}\tr\big(t\mN^{\otimes n}\left(P_\pi\psi_{R^nA^n}P_\pi^*\right)-\mM^{\otimes n}(P_\pi\psi_{R^nA^n}P_\pi^*)\Big)_+\\
 & = \tr\big(t\mN^{\otimes n}(\psi_{R^nA^n})-\mM^{\otimes n}(\psi_{R^nA^n})\big)_+
\end{align}
where the first inequality follows from the data processing inequality of $\tr(\cdot)_+$~\footnote{This can be easily seen from the equation $\tr(X)_+ = (\|X\|_1 + \tr X)/2$ and the data processing inequality of trace norm.}, the first equality follows from the block diagonal structure of $\mN^{\otimes n}\big(\omega_{XR^nA^n}\big)-t\mM^{\otimes n}\big(\omega_{XR^nA^n}\big)$, the second equality follows because $\tr(\cdot)_+$ is unitary invariant and $\cN^{\ox n}, \cM^{\ox n}$ commute with permutations. Together with~\eqref{eq: DH channel variational}, we can conclude that 
\ba
{D_{H}^\eps\left(\mN^{\otimes n}\big\|\mM^{\otimes n}\right)}&\leq-\log\max_{t\in\mbb{R}} \left\{(1-\eps)t-\tr\big(t\mN^{\otimes n}(\varphi_{\tR^nA^n})-\mM^{\otimes n}(\varphi_{\tR^nA^n})\big)_+\right\}\\
&={D_{H}^\eps\left(\mN^{\otimes n}\left(\varphi_{\tR^nA^n}\right)\big\|\mM^{\otimes n}\left(\varphi_{\tR^nA^n}\right)\right)}.
\ea
This completes the proof.
\end{proof} \vspace{0.2cm}

\begin{remark}
We say that a quantum divergence, $\boldD$, satisfies \emph{the direct sum property} if there exists a one-to-one function $f:\mbb{R}_+\to\mbb{R}_+$ such that for any pair of cq-states $\rho,\sigma\in\md(XA)$ of the form $\rho_{XA}\eqdef\sum_xp_x|x\lr x|^X\otimes\rho_x^A$ and $\sigma^{XA}\eqdef\sum_xp_x|x\lr x|^X\otimes\sigma_x^A$
where $\{p_x\}$ is a probability distribution and $\rho_x,\sigma_x\in\md(A)$,
we have $f^{-1}\left(\boldD\left(\rho_{XA}\big\|\sigma_{XA}\right)\right)=\sum_xp_xf^{-1}\left(\boldD\left(\rho_x^A\big\|\sigma_x^A\right)\right)$.
The direct sum property is essentially equivalent to the \emph{general mean property} used by R\'enyi and M\"uller-Lennert et al. for its generalization to the quantum case and holds for almost all the quantum divergences studied in the literature. Following a similar proof, we can show that Lemma~\ref{lem:restrict to symmetric states} holds for any quantum divergence with the direct sum property.
\end{remark}

\end{document}